\documentclass[11pt,a4paper]{article}
\usepackage[english]{babel}
\usepackage{amsmath,amsfonts,amssymb,amsthm}
\usepackage{graphicx}
\usepackage[hidelinks]{hyperref}
\usepackage{cleveref}
\usepackage[dvipsnames,svgnames]{xcolor}
\usepackage{algorithm, algorithmicx, algpseudocode}
\usepackage{pifont}
\usepackage{bbm}
\usepackage{mathtools}
\usepackage{enumerate}
\usepackage{hyphenat}

\usepackage{subcaption}
\usepackage{url}
\usepackage{authblk}

\usepackage[quotation,notion,electronic]{knowledge}
\usepackage{import}
\knowledge{output}{scope=automata,notion}
\knowledge{output}{scope=game,notion}

\knowledge{size}{scope=automata,notion}
\knowledge{size}{scope=memory,notion}

\knowledge{Muller language}[|Muller languages]{scope=language,notion}
\knowledge{Muller condition}[|Muller conditions| Muller]{scope=condition,notion}

\knowledge{accepting}[|rejecting]{scope=cycle,notion}

\knowledge{notion}
 | graph
 | directed graph

\knowledge{notion}
| parity index
| Rabin index

\knowledge{notion}
 | undirected

\knowledge{notion}
 | simple

\knowledge{notion}
 | coloured graph

\knowledge{notion}
 | automaton
 | automata

\knowledge{notion}
 | run over
 | run

\knowledge{notion}
 | graph associated to
 | G(\A )
 | G(\A _c)
 | its graph
 | G(\A')
\knowledge{notion}
 | cycle
 | cycles
 
\knowledge{notion}
 | Rabin condition
 | Rabin
 | Rabin conditions
 | R
\knowledge{notion}
 | Streett condition
 | Streett

\knowledge{notion}
 | parity condition
 | parity

\knowledge{notion}
 | game
 | games

\knowledge{notion}
 | play

\knowledge{notion}
 | $\varepsilon $-free
 | $\varepsilon $-free games
 | $\varepsilon $-free game
 | $\ee $-free

\knowledge{notion}
 | memory structure for the game $\G $
 | memory structure 
 | memories
 | non-chromatic memory
 | memory
 | memory structures

\knowledge{notion}
 | chromatic memory
 | chromatic memories
 | chromatic memory structure
 | chromatic
 | chromatic memory structures
 | Chromatic memories

\knowledge{notion}
 | arena-independent
 | arena-independent memory
 | arena-independent memory for $\WW $-games
 | arena-independent memories for $\WW $-games
 | arena-independent memory for $\F$-games
 | arena-independent memories for $\F$-games
 | Arena-independent memories
 | arena-independent memories

\knowledge{notion}
 | half-positionally determined
 | half-positional
 | half-positionality 

\knowledge{notion}
 | leaf
 | leaves

\knowledge{notion}
 | nodes

\knowledge{notion}
 | Zielonka tree
 | \mathcal {\Z _{\F }} 
 
\knowledge{notion}
 | \mathfrak {m}_{\ZF }
 | \mathfrak{m}_{\ZFi{n}}
 
\knowledge{notion}
 | \mathit{Inf}(w)
 | \mathit{Inf}

\knowledge{notion}
 | language accepted by
 | \L(\A)
 | \L(\A') 
 | \L (\A _1) 
 | \L (\A _2)

\knowledge{notion}
 | \mathit {States}
 | state in common
 | pass through
 | containing

\knowledge{notion}
 | condition of type $X$ on top of $\A $
 | on top of

\knowledge{notion}
 | winning

\knowledge{notion}
 | wins
 | winning strategy
 
\knowledge{notion}
 | strategy
 
\knowledge{notion}
 | sets a winning strategy
 | setting a winning strategy

\knowledge{notion}
 | \mathfrak{mem}_{\mathit{gen}}(\F)
 | \mathfrak {mem}_{\mathit {gen}}(\F _n)
 | memory requirements
 | non-chromatic memory requirements
 | general memory requirements

\knowledge{notion}
 | \mathfrak {rabin}(L_G)
 | \mr (L)
 | \mr(L_G) 
 | \mr (L_\F )

\knowledge{notion}
 | \mathfrak{mem}_{\mathit{chrom}}(\F)
 | chromatic memory requirements

\knowledge{notion}
| \mathfrak{mem}_{\mathit{ind}}(\F)

\knowledge{notion}
 | \mathfrak{mem}_{\mathit{gen}}^{\ee\hyphen \mathrm{free}}(\F)
	
\knowledge{notion}
 | \mathfrak{mem}_{\mathit{chrom}}^{\ee\hyphen \mathrm{free}}(\F)
		
\knowledge{notion}
 | \mathfrak{mem}_{\mathit{ind}}^{\ee\hyphen \mathrm{free}}(\F)

\knowledge{notion}
 | \mathit {Letters}

\knowledge{notion}
 | L_\F
 | L_{\F _n}
 
\knowledge{notion}
 | L_G
 
\knowledge{notion}
| colouring

\knowledge{notion}
 | chromatic number
 | \chi(G)

\knowledge{notion}
 | $\A $-subgraph
 
\knowledge{notion}
 | \mathit {Max\hyphen Priority}
 
\knowledge{notion}
 | complete
 
\knowledge{notion}
 | $C$-Strongly Connected Component
 | $C$-SCC
 | $C_i$-SCC
 
\knowledge{notion}
 | \mathtt {SCC\hyphen Decomposition}
 
 \knowledge{notion}
  | \mathtt {Complete\hyphen SCC}
  | \mathtt {Complete \hyphen SCC} 

 \knowledge{notion}  
  | \mathtt {MaxInclusion}
  
\knowledge{notion}   
 | \iota
 
\knowledge{notion}    
 | $\GG $-labelled-tree 
 | tree
   
\knowledge{notion}       
 | equivalent
 
\knowledge{notion}    
 | \gg (\ell )
 | \gg (\ell _i)
 
\knowledge{notion} 
 | adheres   
 
\knowledge{notion} 
 | strongly connected

\knowledge{notion} 
 | strongly connected component
 | strongly connected components 
 
\knowledge{notion} 
 | ergodic
 
\knowledge{notion} 
| generalised Büchi
| generalised Büchi condition

\knowledge{notion} 
| generalised co-Büchi
| generalised co-Büchi condition

\usepackage{tikz}
\usetikzlibrary{automata, positioning, arrows,arrows.meta,calc,decorations.markings,math,shapes.geometric}
\tikzset{
	>=stealth',
	-={stealth',ultra thick,scale=3} 
	node distance=1cm, 
	every state/.style={thick, fill=gray!10}, 
	initial text=$ $, 
}
\newtheorem{definition}{Definition}
\newtheorem{theorem}[definition]{Theorem}
\newtheorem{proposition}[definition]{Proposition}
\newtheorem{lemma}[definition]{Lemma}
\newtheorem{corollary}[definition]{Corollary}
\newtheorem{fact}{Fact}

\newtheorem{remark}[definition]{Remark}
\newtheorem{example}[definition]{Example}

\let\ab\allowbreak
\mathchardef\hyphen=45 

\definecolor{Green2}{HTML}{3EA514}
\definecolor{Red2}{HTML}{FF0400}
\definecolor{Orange2}{HTML}{E6670A}
\definecolor{Violet2}{HTML}{CE1ff9}

\DeclareMathAlphabet{\mathpzc}{OT1}{pzc}{m}{it}

\newcommand{\NN}{\mathbb{N} }

\newcommand{\WW}{\mathbb{W}}

\newcommand{\F}{\mathcal{F}}
\newcommand{\G}{\mathcal{G}}

\renewcommand{\L}{\mathcal{L}}
\newcommand{\M}{\mathcal{M}}

\newcommand{\A}{\mathcal{A}}
\newcommand{\B}{\mathcal{B}}
\newcommand{\Z}{\mathcal{Z}}

\renewcommand{\P}{\mathcal{P}}

\renewcommand{\O}{\mathcal{O}}

\renewcommand{\S}{\mathcal{S}}

\newcommand{\dd}{\delta}

\renewcommand{\ss}{\sigma}

\newcommand{\rr}{\varrho}

\renewcommand{\aa}{\alpha}

\renewcommand{\tt}{\tau}

\renewcommand{\SS}{\Sigma}
\newcommand{\GG}{\Gamma}
\newcommand{\oo}{\omega}

\renewcommand{\gg}{\gamma}
\newcommand{\ee}{\varepsilon}

\newcommand{\ZF}{"\mathcal{\Z_{\F}}"}
\newcommand{\ZFi}[1]{\Z_{\F_{#1}}}
\newcommand{\mF}{"\mathfrak{mem}_{\mathit{gen}}(\F)"}
\newcommand{\mFn}{"\mathfrak{mem}_{\mathit{gen}}(\F_n)"}
\newcommand{\cmF}{"\mathfrak{mem}_{\mathit{chrom}}(\F)"}

\newcommand{\mr}{\mathfrak{rabin}}
\newcommand{\rlg}{"\mathfrak{rabin}(L_G)"}

\newcommand{\indmF}{"\mathfrak{mem}_{\mathit{ind}}(\F)"}

\newcommand{\mFfree}{"\mathfrak{mem}_{\mathit{gen}}^{\ee\hyphen \mathrm{free}}(\F)"}
\newcommand{\cmFfree}{"\mathfrak{mem}_{\mathit{chrom}}^{\ee\hyphen \mathrm{free}}(\F)"}
\newcommand{\indmFfree}{"\mathfrak{mem}_{\mathit{ind}}^{\ee\hyphen \mathrm{free}}(\F)"}

\newcommand{\minf}{"\mathit{Inf}"}

\newcommand{\macc}{\mathit{Acc}}

\newcommand{\mletters}{"\mathit{Letters}"}
\newcommand{\maxpr}{"\mathit{Max\hyphen Priority}"}

\newcommand{\mstates}{"\mathit{States}"}
\newcommand{\mnodes}{\mathit{Nodes}}

\newcommand{\tinput}{\mathtt{Input}}

\newcommand{\NP}{\mathtt{NP}}
\newcommand{\NPc}{\mathtt{NP}\hyphen\text{complete}}
\newcommand{\CN}{\mathtt{Chromatic} \hyphen \mathtt{Number}}
\newcommand{\nextmove}{\mathtt{next} \hyphen \mathtt{move}}

	\title{On the Minimisation of Transition-Based Rabin Automata and the Chromatic Memory Requirements of Muller Conditions}
\author{Antonio Casares\thanks{This work was done while the author was participating in the program \emph{Theoretical Foundations of Computer Systems} at the Simons Institute for the Theory of Computing.}}
\affil[]{LaBRI, Universit{\'e} de Bordeaux, France\\ \texttt{antonio.casares-santos@labri.fr}}
\date{}
\begin{document}

	\maketitle

	\begin{abstract}
		In this paper, we relate the problem of determining the chromatic memory requirements of Muller conditions with the minimisation of transition-based Rabin automata. Our first contribution is a proof of the $\NP$-completeness of the minimisation of transition-based Rabin automata. Our second contribution concerns the memory requirements of games over graphs using  Muller conditions. A memory structure is a finite state machine that implements a strategy and is updated after reading the edges of the game; the special case of chromatic memories being those structures whose update function only consider the colours of the edges. We prove that the minimal amount of chromatic memory required in games using a given Muller condition is exactly the size of a minimal Rabin automaton recognising this condition. Combining these two results, we deduce that finding the chromatic memory requirements of a Muller condition is $\NPc$. This characterisation also allows us to prove that chromatic memories cannot be optimal in general, disproving a conjecture by Kopczyński.

	\end{abstract}
	
	\section{Introduction}\label{section : Intro}
	\subparagraph*{Games and memory.}Automata on infinite words and infinite duration games over graphs are well established areas of study in Computer Science, being central tools used to solve problems such as the synthesis of reactive systems (see for example the Handbook \cite{PitermanPnueli2018Handbook}). Games over graphs are used to model the interaction between a system and the environment, and winning strategies can be used to synthesize controllers ensuring that the system satisfies some given specification. 
The games we will consider are played between two players (Eve and Adam), that alternatively move a pebble through the edges of a graph forming an infinite path. In order to define which paths are winning for the first player, Eve, we suppose that each transition in the game produces a colour in a set $\GG$, and a winning condition is defined by a subset $\WW\subseteq \GG^\oo$. 
A fundamental parameter of the different winning conditions is the amount of memory that the players may require in order to define a winning strategy in games where they can force a victory. 
This parameter will influence the complexity of algorithms solving games that use a given winning condition, as well as the resources needed in a practical implementation of such a strategy as a controller for a reactive system.

A "memory structure" for Eve for a given game is a finite state machine that implements a strategy: for every position of the game, each state of the memory determines what move to perform next. After a transition of the game takes place, the memory state is updated according to an update function. We consider $3$ types of memory structures:
\begin{itemize}
	\item General memories.
	\item "Chromatic memories": if the update function only takes as input the colour produced by the transition of the game.
	\item "Arena-independent memories" for a condition $\WW$: if the memory structure can be used to implement winning strategies in any game using the condition $\WW$.
\end{itemize}

In this work, we study these three notions of memories for \kl(condition){Muller conditions}, an important class of winning conditions that can be used to represent any $\oo$-regular language via some deterministic "automaton". 
Muller conditions appear naturally, for example, in the synthesis of reactive systems specified in Linear Temporal Logic \cite{PR89Synthesis, MullerSickert17LTLtoDeterministic}.

In the seminal paper \cite{DJW1997memory}, the authors establish the exact general memory requirements of Muller conditions, giving matching upper and lower bounds for every Muller condition in terms of its "Zielonka tree". However, the memory structures giving the upper bounds are not "chromatic". In his PhD thesis \cite{Kopczynski2006Half,Kopczynski2008PhD}, Kopczyński raised the questions of whether minimal "memory structures" for a given game can be chosen to be "chromatic", and whether "arena-independent memories" can be optimal, that is, if for each condition $\WW$ there is a game won by Eve where the optimal amount of memory she can use is the size of a minimal "arena-independent memory" for $\WW$. Another question appearing in \cite{Kopczynski2006Half,Kopczynski2008PhD} concerns the influence in the memory requirements of allowing or not $\ee$-transitions in games (that is, transitions that do not produce any colour). In particular, Kopczyński asks whether all conditions that are "half-positionally determined" over transition-coloured games without $\varepsilon$-transitions are also half-positionally determined when allowing $\varepsilon$-transitions (it was already shown in~\cite{Zielonka1998infinite} that it is not the case in state-coloured games).

In this work, we characterise the minimal amount of "chromatic memory" required by Eve in games using a \kl(condition){Muller condition} as the size of a minimal deterministic transition-based "Rabin" automaton recognising the Muller condition, that can also be used as an "arena-independent" memory (Theorem~\ref{theorem : EquivRabin-ChromaticMemory}); further motivating the study of the minimisation of transition-based Rabin automata. We prove that, in general, this quantity is strictly greater than the "general memory requirements" of the Muller condition, answering negatively the question by Kopczyński (Proposition~\ref{prop : chromatic_strictly-greater}). Moreover, we show that the "general memory requirements" of a Muller condition are different over "$\ee$-free" games and over games with $\ee$-transitions (Proposition~\ref{prop : epsilon-more-memory}), but that this is no longer the case when considering the "chromatic memory requirements" (Theorem~\ref{theorem : EquivRabin-ChromaticMemory}). In particular, in order to obtain the lower bounds of \cite{DJW1997memory} we need to use $\varepsilon$-transitions. However, the question stated in \cite{Kopczynski2006Half,Kopczynski2008PhD} of whether allowing $\ee$-transitions could have an impact on the "half-positionality" of conditions remains open, since it cannot be the case for \kl(condition){Muller conditions} (Lemma~\ref{lemma : positionalRabin}).

\subparagraph*{Minimisation of transition-based automata.}

Minimisation is a well studied problem for many classes of automata. Automata over finite words can be minimised in polynomial time \cite{Hopcroft71nLognAlgorithm}, and for every regular language there is a canonical minimal automaton recognising it. For automata over infinite words, the status of the minimisation problem for different models of $\oo$-automata is less well understood. Traditionally, the acceptance conditions of $\oo$-automata have been defined over the set of states; however, the use of transition-based automata is becoming common in both practical and theoretical applications (see for instance \cite{GiannakopoulouLerda2002Transitions}), and there is evidence that decision problems relating to transition-based models might be easier than the corresponding problems for state-based ones.
The minimisation of state-based Büchi automata has been proven to be $\NPc$ by Schewe (therefore implying the $\NP$-hardness of the minimisation of state-based "parity", "Rabin" and "Streett" automata), both for deterministic \cite{Schewe10MinimisingNPComplete} and Good-For-Games (GFG) automata \cite{Schewe20MinimisingGFG}. However, these reductions strongly use the fact that the acceptance condition is defined over the states and not over the transitions. Abu Radi and Kupferman have proven that the minimisation of GFG-transition-based co-Büchi automata can be done in polynomial time and that a canonical minimal GFG-transition-based automaton can be defined for co-Büchi languages \cite{AbuRadiKupferman19Minimizing, AbuRadiKupferman20Canonicity}. 
This suggests that transition-based automata might be a more adequate model for $\oo$-automata, raising many questions about the minimisation of different kinds of transition-based automata (Büchi, "parity", "Rabin", GFG-parity, etc).
Moreover, Rabin automata are of great interest, since the determinization of Büchi automata via Safra's construction naturally provides deterministic transition-based Rabin automata \cite{Safra1988onthecomplexity, Schewe2009tighter}, and, as proven in Theorem~\ref{theorem : EquivRabin-ChromaticMemory}, these automata provide minimal "arena-independent memories" for Muller games.

In Section \ref{subsec: MinimisationRabin}, we prove that the minimisation of transition-based "Rabin" automata is $\NPc$ (Theorem~\ref{theorem: MinimisingRabinNP-comp}). The proof consists in a reduction from the "chromatic number" problem of graphs. This reduction uses a particularly simple family of $\oo$-regular languages: languages $L\subseteq \SS^\oo$ that correspond to \kl(condition){Muller conditions}, that is, whether a word $w\in \SS^\oo$ belongs to $L$ or not only depends in the set of letters appearing infinitely often in $w$ (we called them \kl(language){Muller languages}). 
A natural question is whether we can extend this reduction to prove the $\NP$-hardness of the minimisation of other kinds of transition-based automata, like "parity" or "generalised Büchi" ones. However, we prove in Section~\ref{subsec: MinimisationParity} that the minimisation of "parity" and "generalised Büchi" automata recognising \kl(language){Muller languages} can be done in polynomial time. This is based in the fact that the minimal parity automaton recognising a \kl(condition){Muller condition} is given by the "Zielonka tree" of the condition \cite{CCF21Optimal, MeyerSickert21OptimalPractical}.\\

These results allow us to conclude that determining the "chromatic memory requirements" of a \kl(condition){Muller condition} is $\NPc$ even if the condition is represented by its "Zielonka tree" (Theorem~\ref{theorem : ChromaticMemoryNP-hard}). This is a surprising result, since the "Zielonka tree" of a \kl(condition){Muller condition} allows to compute in linear time the "non-chromatic memory requirements" of it \cite{DJW1997memory}.

\subparagraph*{Related work.}

As already mentioned, the works \cite{DJW1997memory, Kopczynski2006Half, Zielonka1998infinite} extensively study the "memory requirements" of \kl(condition){Muller conditions}. In the paper \cite{ColcombetN2006PositionalEdge}, the authors characterise "parity" conditions as the only prefix-independent conditions that admit positional strategies over transition-co\-lou\-red infinite graphs. This characterisation does not apply to state-coloured games, which supports the idea that transition-based systems might present more canonical properties. Conditions that admit "arena-independent" memories are characterised in \cite{BRORV20FiniteMemory}, extending the work of \cite{GimbertZielonka2005Memory} characterising conditions that accept positional strategies over finite games. The "memory requirements" of generalised safety conditions have been established in \cite{ColcombetFH14PlayingSafe}.
The use of Rabin automata as memories for games with $\oo$-regular conditions have been fruitfully used in \cite{Colcombetz2009tight} in order to obtain theoretical lower bounds on the size of deterministic Rabin automata obtained by the process of determinisation of Büchi automata.

Concerning the minimisation of automata over infinite words, beside the aforementioned results of \cite{Schewe10MinimisingNPComplete, Schewe20MinimisingGFG, AbuRadiKupferman19Minimizing}, it is also known that weak automata can be minimised in $\O(n\log n)$ \cite{Loding01WeakMinimisation}. The algorithm minimising a "parity" "automaton" recognising a \kl(language){Muller language} presented in the proof of Proposition~\ref{prop : MinimisationParityPoly} can be seen as a generalisation of the algorithm appearing in \cite{CartonMaceiras99RabinIndex} computing the Rabin index of a parity automaton. Both of them have their roots in the work of Wagner \cite{Wagner1979omega}.

\paragraph*{Organisation of this paper.}

In Section~\ref{section : MinimisingRabin} we discuss the minimisation of transition-based Rabin and parity automata. We give the necessary definitions in Section~\ref{subsection : AutomataDefinitions}, in Section~\ref{subsec: MinimisationRabin} we show the $\NP$-completeness of the minimisation of Rabin automata and in Section~\ref{subsec: MinimisationParity} we prove that we can minimise transition-based parity and generalised Büchi automata recognising Muller languages in polynomial time.

In Section~\ref{section : MemoryDefinitions} we introduce the definitions of games and memory structures, and we discuss the impact on the memory requirements of allowing or not $\ee$-transitions in the games.

In Section~\ref{section : ChromaticMemory}, the main contributions concerning the chromatic memory requirements of Muller conditions are presented.

%

	\section{Minimising transition-based automata}\label{section : MinimisingRabin}
	In this section, we present our main contributions concerning the minimisation of transition-based automata. We start in Section~\ref{subsection : AutomataDefinitions} by giving some basic definitions and results related to automata used throughout the paper. In Section~\ref{subsec: MinimisationRabin} we show a reduction from the 
problem of determining the "chromatic number" of a graph to the minimisation of "Rabin" automata, proving the $\NP$-completeness of the latter. Moreover, the languages used in this proof are \kl(language){Muller languages}. In Section~\ref{subsec: MinimisationParity}  we prove that, on the contrary, we can minimise "parity" and "generalised Büchi" automata recognising Muller languages in polynomial time.

\subsection{Automata over infinite words}\label{subsection : AutomataDefinitions}

\subsubsection*{General notations}
The greek letter $\oo$ stands for the set $\{0,1,2,\dots\}$. Given a set $A$, we write $\P(A)$ to denote its power set and $|A|$ to denote its cardinality. A word over an alphabet $\SS$ is a sequence of letters from $\SS$. We let $\SS^*$ and $\SS^\oo$ be the set of finite and infinite words over $\SS$, respectively. For an infinite word $w\in \SS^\oo$, we write $\AP""\mathit{Inf}(w)""$ to denote the set of letters that appear infinitely often in $w$. We will extend functions $\gg: A \rightarrow \GG$ to $A^*$, $A^\oo$ and $\P(A)$ in the natural way, without explicitly stating it.

A (directed)  \AP""graph"" $G=(V,E)$ is given by a set of vertices $V$ and a set of edges $E\subseteq V\times V$. A graph $G=(V,E)$ is \AP""undirected"" if every pair of vertices $(v,u)$ verifies $(v,u)\in E \; \Leftrightarrow \; (u,v)\in E$. A graph $G=(V,E)$ is \AP""simple"" if $(v,v)\notin E$ for any $v\in V$. 
A subgraph $\S$ of $G$ is \AP""strongly connected"" if for every pair of vertices $v_1, v_2$ in $\S$, there is a path in $\S$ from $v_1$ to $v_2$. A \AP""strongly connected component"" of $G$ is a maximal strongly connected subgraph of $G$. We say that a "strongly connected component" $\S$ is \AP""ergodic"" if no vertex of $\S$ has outgoing edges leading to vertices not in $\S$.

A \AP""coloured graph"" $G=(V,E)$ is given by a set of vertices $V$ and a set of edges $E\subseteq V\times C_1\times \dots\times C_k \times V$, where $C_1,\dots , C_k$ are sets of colours.

\subsubsection*{Automata}
An \AP""automaton"" is a tuple $\A=(Q,\Sigma, q_0, \dd, \GG, \macc)$, where $Q$ is a finite set of states, $\SS$ is a finite input alphabet, $q_0\in Q$ is an initial state, $\dd: Q \times \SS \rightarrow Q \times \GG$ is a transition function, $\GG$ is an output alphabet and $\macc$ is an accepting condition defining a subset $\WW \subseteq \GG^\oo$ (the conditions will be defined more precisely in the next paragraph). In this paper, all automata will be deterministic and complete (that is, $\dd$ is a function) and transition-based (that is, the output letter that is produced depends on the transition, and not only on the arrival state). The \intro(automata){size} of an automaton is its number of states, $|Q|$.

Given an input word $w=w_0w_1w_2\dots \in \SS^\oo$, the \AP""run over"" $w$ in $\A$ is the only sequence of pairs  $(q_0,c_0),(q_1,c_1),\dots \in Q\times \GG$ verifying that $q_0$ is the initial state and $\dd(q_i,w_i)=(q_{i+1},c_i)$. The \AP\intro(automata){output} produced by $w$ is the word $c_0c_1c_2\dots \in \GG^\oo$. A word $w\in \SS^\oo$ is \emph{accepted by} the automaton $\A$ if its output belongs to the set $\WW\subseteq \GG^\oo$ defined by the accepting condition. The \AP""language accepted by"" an automaton $\A$, written $\L(\A)$, is the set of words accepted by $\A$. Given two automata $\A$ and $\B$ over the same input alphabet $\SS$, we say that they are \AP""equivalent"" if $\L(\A)=\L(\B)$.

Given an automaton $\A$, the \AP""graph associated to"" $\A$, denoted $G(\A)$, is the "coloured graph" $G(\A)= (Q, E_\A)$, whose set of vertices is $Q$, and the set of edges $E_\A\subset Q\times \SS \times \GG \times Q$ is given by $(q,a,c,q')\in E_\A$ if $\dd(q,a)=(q',c)$. We denote by $\AP""\iota""\colon E_\A \to \SS$ the projection over the second component and by $\gg\colon E_\A \to \GG$ the projection over the third one.

A \AP""cycle"" of an automaton $\A$ is a subset of edges $\ell\subseteq E_\A$ such that there is a state $q\in Q$ and a path in $"G(\A)"$ starting and ending in $q$ passing through exactly the edges in $\ell$. We write $\AP""\gg(\ell)""=\bigcup_{e\in \ell}\gg(e)$ to denote the set of colours appearing in the cycle $\ell$. A state $q\in Q$ is contained in a cycle  $\ell\subseteq E_\A$ if there is some edge in $\ell$ whose first component is $q$. We write $\AP""\mathit{States}""(\ell)$ to denote the set of states contained in $\ell$.

\subsubsection*{Acceptance conditions}

Let $\GG$ be a set of colours. We define next some of the acceptance conditions used to define subsets $\WW \subseteq \GG^\oo$. All the subsequent conditions verify that the acceptance of a word $u\in \GG^\oo$ only depends on the set $\minf(u)$.

\begin{description}
	
	\item[Muller.] A  \AP\intro(condition){Muller condition} is given by a family of subsets  $\F=\{S_1,\dots,S_k\}$, $S_i\subseteq \Gamma$. A word $u\in \Gamma^\oo$ is accepting if $\minf(u)\in \F$.
	
	\item[Rabin.] A  \AP""Rabin condition"" is represented by a family of \emph{Rabin pairs}, $R=\ab\{(E_1,F_1),\dots,\ab(E_r,F_r)\}$, where $E_i,F_i\subseteq \Gamma$. A word $u\in \Gamma^\oo$ is accepting if $\minf(u) \cap E_i \neq \emptyset $ and $ \minf(u) \cap F_i = \emptyset $ for some index $i\in \{1,\dots,r\}$.
	
	\item[Streett.] A  \AP""Streett condition"" is represented by a family of pairs $S=\{(E_1,F_1),\dots,\ab (E_r,F_r)\}$, $E_i,F_i\subseteq \Gamma$. A word $u\in \Gamma^\oo$ is accepting if $\minf(u) \cap E_i \neq \emptyset \; \rightarrow \; \minf(u) \cap F_i \neq \emptyset$ for every $i\in \{1,\dots,r\}$.
	
	
	\item[Parity.] To define a  \AP""parity condition"" we suppose that $\Gamma$ is a finite subset of $\NN$. A word $u\in \Gamma^\oo$ is accepting if $\max  \minf(u) \text{ is even}$.
	The elements of $\Gamma$ are called  \emph{priorities} in this case.

	\item[Generalised Büchi.] A  \AP""generalised Büchi condition"" is represented by a family of subsets $\{B_1, \dots, B_r\}$, $B_i \subseteq \Gamma$. A word $u\in \Gamma^\oo$ is accepted if $\minf(u)\cap B_i \neq \emptyset$  for all $i\in \{1,\dots,r\}$.
	
	\item[Generalised co-Büchi.] A  \AP""generalised co-Büchi condition"" is represented by a family of subsets $\{B_1, \dots, B_r\}$, $B_i \subseteq \Gamma$. A word $u\in \Gamma^\oo$ is accepted if $\minf(u)\cap B_i = \emptyset$  for some $i\in \{1,\dots,r\}$.
\end{description}

An "automaton" $\A$ using a condition of type $X$ will be called an $X$-automaton.

We remark that all the previous conditions define a family of subsets $\F \subseteq \P(\Gamma)$ and can therefore be represented as Muller conditions (in particular, all automata referred to in this paper can be regarded as Muller automata). Also, parity conditions can be represented as Rabin or Streett ones. We say that a language $L\subseteq \GG^\oo$ is a \AP\intro(language){Muller language} if $u_1\in L$ and $u_2\notin L$ implies that $\minf(u_1)\neq \minf(u_2)$. We associate to each Muller condition $\F$ the language $\AP""L_\F""=\{w\in \GG^\oo \: : \: \minf(w)\in \F\}$.

The \AP""parity index"" (also called \emph{Rabin index}) of an $\oo$-regular language $L\subseteq \SS^\oo$ is the minimal $p\in \NN$ such that there exists a deterministic parity automaton recognising $L$ using $p$ priorities in its parity condition.

Given an $\oo$-regular language $L\subseteq \SS^\oo$, we write $\AP""\mr(L)""$ to denote the \kl(automata){size} of a minimal Rabin automaton recognising $L$.

\begin{remark}\label{remark:RabinAndStreett}
	Let $\A$ be a "Rabin"-automaton recognising a language $L\subseteq \SS^\oo$ using Rabin pairs $R=\ab\{(E_1,F_1),\dots,\ab(E_r,F_r)\}$. If we consider the "Streett" automaton obtained by setting the pairs of $R$ as defining a Streett condition over the structure of $\A$, we obtain a Streett automaton $\A'$ recognising the language $\SS^\oo\setminus L$ (and vice versa). Therefore, the size of a minimal Rabin automaton recognising $L$ coincides with that of a minimal Streett automaton recognising $\SS^\oo\setminus L$, and the minimisation problem for both classes of automata is equivalent. Similarly for "generalised Büchi" and "generalised co-Büchi" automata.
\end{remark}

Let $\A$ be an "automaton" using some of the acceptance conditions above defining a family $\F \subseteq \P(\Gamma)$. We remark that since the acceptance of a "run" only depends on the set of colours produced infinitely often, we can associate to each "cycle" $\ell$ of $\A$ an accepting or rejecting status. We say that a "cycle" $\ell$ of $\A$ is \intro(cycle){accepting} if $"\gg(\ell)"\in \F$ and that it is \emph{rejecting} otherwise.

We are going to be interested in simplifying the acceptance conditions of automata, while preserving their structure. We say that we can \emph{define a condition of type $X$ on top of a Muller automaton $\A$} if we can recolour the transitions of $\A$ with colours in a set $\GG'$ and define a condition of type $X$ over $\GG'$ such that the resulting automaton is equivalent to $\A$. Definition~\ref{Def: ConditionXOnTopOf} formalises this notion.

\begin{definition}\label{Def: ConditionXOnTopOf}
	Let $X$ be some of the conditions defined previously and let $\A=(Q,\Sigma, q_0, \dd, \GG, \F)$ be a Muller automaton. We say that we can define a \AP""condition of type $X$ on top of $\A$"" if there is an $X$-condition over a set of colours $\GG'$ and an automaton $\A'=(Q,\Sigma, q_0, \dd', \GG', X)$ verifying:
		\begin{itemize}
		\item $\A$ and $\A'$ have the same set of states and the same initial state.
		\item $\dd(q,a)=(p,c) \; \Rightarrow \dd'(q,a)=(p,c')$, for some $c'\in \GG'$, for every $q\in Q$ and  $a\in \SS$ (that is, $\A$ and $\A'$ have the same transitions, except for the colours produced).
		\item $"\L(\A)"= "\L(\A')"$.
	\end{itemize}
\end{definition}

The next proposition, proven in \cite{CCF21Optimal}, characterises automata that admit "Rabin" conditions "on top of" them. It will be a key property used throughout the paper.

\begin{proposition}[\cite{CCF21Optimal}]\label{prop: Rabin iff cycles}
	Let $\A=(Q,\Sigma, q_0, \dd, \GG, \F)$ be a Muller automaton. The following properties are equivalent:
	\begin{enumerate}
		\item We can define a "Rabin" condition "on top of" $\A$.
		\item Any pair of "cycles" $\ell_1$ and $\ell_2$ in $\A$ verifying $\mstates(\ell_1)\cap \mstates(\ell_2)\neq \emptyset$ satisfies that if both $\ell_1$ and $\ell_2$ are \kl(cycle){rejecting}, then $\ell_1 \cup \ell_2$ is also a rejecting cycle.
	\end{enumerate}
\end{proposition}

\subsubsection*{The Zielonka tree of a Muller condition}
In order to study the memory requirements of Muller conditions, Zielonka introduced in \cite{Zielonka1998infinite} the notion of split trees (later called Zielonka trees) of Muller conditions. The Zielonka tree of a Muller condition naturally provides a minimal parity automaton recognising this condition \cite{CCF21Optimal, MeyerSickert21OptimalPractical}. We will use this property to show that parity automata recognising Muller languages can be minimised in polynomial time in Proposition~\ref{prop : MinimisationParityPoly}. We will come back to Zielonka trees in Section \ref{section : ChromaticMemory} to discuss the characterisation of the memory requirements of Muller conditions.

\begin{definition}
	Let $\GG$ be a set of labels. We give the definition of a \AP""$\GG$-labelled-tree"" by induction:
	\begin{itemize}
		\item $T=\langle A, \langle \emptyset \rangle \rangle$ is a $\GG$-labelled-tree for any $A\subseteq \GG$. In this case, we say that $T$ is a \AP""leaf"" and $A$ is its label.
		\item If $T_1, \dots, T_n$ are $\GG$-labelled-trees, then $T=\langle A, \langle T_1, \dots, T_n\rangle\rangle$ is a $\GG$-labelled-tree for any $A\subseteq \GG$.	In that case, we say that $A$ is the label of $T$ and $T_1, \dots, T_n$ are their children.		
	\end{itemize}
	The set of \AP""nodes"" of a tree $T$ is defined recursively as:
	\[ \mnodes(T)=\{T\} \cup \bigcup\limits_{T' \text{ child of } T} \mnodes(T'). \]
\end{definition} 

\begin{definition}[\cite{Zielonka1998infinite}]\label{def: ZielonkaTree}
	Let $\F\subseteq \P(\GG)$ be a \kl(condition){Muller condition}. The \AP""Zielonka tree"" of $\F$, denoted $\ZF$, is the $\GG$-labelled-tree defined recursively as follows: let $A_1, \dots A_k$ be the maximal subsets of $\GG$ such that $A_i\in \F \; \Leftrightarrow \; \GG \notin \F$.
	\begin{itemize}
		\item If no such subset $A_i\subseteq \GG$ exists, then $\ZF=\langle \GG, \langle \emptyset \rangle \rangle$.
		\item Otherwise, $\ZF=\langle \GG, \langle  \ZFi{1}, \dots , \ZFi{k}\rangle \rangle$, where $\ZFi{i}$ is the Zielonka tree for the condition $\F_i= \F \cap \P(A_i)$ over the colours $A_i$.
	\end{itemize}	
\end{definition}
An example of a Zielonka tree can be found in Figure~\ref{Fig : ZielonkaTree} (page 28).

\begin{proposition}[\cite{CCF21Optimal}, \cite{MeyerSickert21OptimalPractical}]\label{prop: minimalParityZielonka}
	Let $\F$ be a \kl(condition){Muller condition} and $\ZF$ its Zielonka tree. We can build in linear time in the representation of $\ZF$
	a "parity" "automaton" recognising $"L_\F"$ that has as set of states the leaves of $\ZF$. This automaton is minimal, that is, any other parity automaton recognising $"L_\F"$ has at least as many states as the number of leaves of $\ZF$. 
\end{proposition}
	\subsection{Minimising transition-based Rabin and Streett automata is NP-complete}\label{subsec: MinimisationRabin}
	
	This section is devoted to proving the $\NP$-completeness of the minimisation of transition-based "Rabin" automata, stated in Theorem \ref{theorem: MinimisingRabinNP-comp}.\\

	For the containment in $\NP$, we use the fact that we can test language equivalence of Rabin automata in polynomial time. 
	
	\begin{proposition}[\cite{ClarkeDK93Unified}]\label{Prop: equivalence Rabin in P}
	Let $\A_1$ and $\A_2$ be two "Rabin" "automata" over $\SS$. We can decide in polynomial time on the representation of the automata if $"\L(\A_1)"="\L(\A_2)"$. (We recall that all considered automata are deterministic).
	\end{proposition}

	\begin{corollary}
	Given a "Rabin" "automaton" $\A$ and a positive integer $k$, we can decide in non-deterministic polynomial time whether there is an "equivalent" Rabin automaton of size $k$.
	\end{corollary}
	\begin{proof}
		A non-deterministic Turing machine just has to guess an equivalent automaton $\A_k$ of size $k$, and by Proposition~\ref{Prop: equivalence Rabin in P} it can check in polynomial time whether $\L(\A)=\L(\A_k)$.
	\end{proof}
	
	In order to prove the $\NP$-hardness, we will describe a reduction from the $\CN$ problem (one of 21 Karp's $\NP$-complete problems) to the minimisation of transition-based Rabin automata. Moreover, this reduction will only use languages that are \kl(language){Muller languages} of Rabin index 3.\\
	
	\begin{definition}
		Let $G=(V,E)$ be a "simple" "undirected" "graph". A \AP""colouring"" of size $k$ of $G$ is a function $c: V \rightarrow [1,k]$ such that for any pair of vertices $v, v' \in V$, if $(v,v')\in E$ then $c(v)\neq c(v')$.
		
		The \AP""chromatic number"" of a simple undirected $G$, written $\chi(G)$, is the smallest number $k$ such that there exists a colouring of size $k$ of $G$.
	\end{definition}

	\begin{lemma}[\cite{Karp72Reducibility}]
		Deciding whether a "simple" "undirected" "graph" has a "colouring" of size $k$ is $\NPc$.
	\end{lemma}

	Let $G=(V,E)$ be a simple undirected "graph", $n$ be its number of vertices and $m$ its number of edges. 
	 We consider the language $\AP""L_G""$ over the alphabet $V$ given by:
	\[ L_G = \bigcup\limits_{(v,u)\in E} V^*(v^+u^+)^\oo.\]
	
	
	That is, a sequence $w\in V^\oo$ is in $L_G$ if eventually it alternates between exactly two vertices connected by an edge in $G$.
	
	\begin{remark}
		For any "simple" "undirected" "graph" $G$, $"L_G"$ is a \kl(language){Muller language} over $V$, that is, whether a word $w\in V^\oo$ belongs to $"L_G"$ or not only depends on $\minf(w)$. Moreover, the parity index of this condition is at most $[1 \hyphen 3]$.
	\end{remark}
	However, we cannot extend this reduction to show $\NP$-hardness of the minimisation of transition-based "parity" automata, as we will show in Section \ref{subsec: MinimisationParity} that we can minimise "parity" automata recognising \kl(language){Muller languages} in polynomial time.
	
	\begin{lemma}
		We can build a "Rabin" automaton $\A_G$ of \kl(automata){size} $n$ over the alphabet $V$ recognising $"L_G"$ in $\O(mn^2)$.
	\end{lemma}
	\begin{proof}
		We define the automaton $\A_G=(Q,V, q_0, \dd, V\times V, R)$ as follows:
		\begin{itemize}
			\item $Q=V$.
			\item $q_0$ an arbitrary vertex in $Q$.
			\item $\dd(v,x)=(x,(v,x))$, for $v, x \in V$.
			\item $R=\{ (E_{(v,u)}, F_{(v,u)}) \: : \: {(v,u)\in E} \}$, where we define for each $(v,u)\in E$ the sets $E_{(v,u)}, F_{(v,u)}$ as:
			\begin{itemize}
				\item $E_{(v,u)}=\{(v,u)\}$
				\item $F_{(v,u)} = V\times V \setminus \{(v,u), (u,v), (v,v), (u,u)\}$.
			\end{itemize}
%
		\end{itemize}
	
	That is, the states of the automaton are the vertices of the graph $G$, and when we read a letter $u\in V$ we jump to the state $u$. The colours defining the "Rabin condition" are all pairs of vertices, and we define one Rabin pair for each edge of the graph. This Rabin pair will enable to accept words that eventually alternate between the endpoints of the edge.

	We prove that $"L_G"= \L(\A_G)$. If $w\in "L_G"$, the word $w$ is eventually of the form $(v^+u^+)^\oo$ for some $v,u\in V$ such that $(v,u)\in E$, so eventually we will only visit the states $v$ and $u$ of the automaton. If we are in $v$ and read the letter $v$, we produce the colour $(v,v)$ that is not contained in the set $F_{(v,u)}$. If we read the letter $u$, we will produce $(v,u)$, contained in $E_{(v,u)}$ and not appearing in $F_{(v,u)}$ nor in $F_{(u,v)}$. 
	The behaviour is symmetric from the state $u$. Let us denote $\aa\in E^\oo$ the word \kl(automata){output} by the automaton reading $w$. We obtain that $\minf(\aa)\cap E_{(v,u)} \neq \emptyset$ and $\minf(\aa)\cap F_{(v,u)} = \emptyset$, so the word $w$ is accepted by the automaton.
	
	Conversely, if a word $w$ is accepted by $\A_G$, then the \kl(automata){output} $\aa\in (V\times V)^\oo$ must verify $\minf(\aa)\cap F_{(v,u)} = \emptyset$ for some $(v,u)\in E$, so eventually $\aa$ only contains the pairs $(v,u), (u,v), (v,v)$ and $(u,u)$. If $w$ was eventually of the form $v^\oo$ or $u^\oo$, we would have that $\minf(\aa)\cap E_{(v,u)} = \emptyset$  for all $(v,u)\in E$. We conclude that the word $w$ eventually alternates between two vertices connected by an edge, so $w\in "L_G"$. 
	
	The automaton has $n$ states, the transition function $\dd$ has size $\O(n^2)$ and the Rabin condition $R$ has size $\O(mn^2)$.
	\end{proof}

	\begin{lemma}\label{lemma : RabinForLG = ChromaticNumberG}
		Let $G=(V,E)$ be a "simple" "undirected" "graph". Then, the \kl(automata){size} of a minimal "Rabin" automaton recognising $"L_G"$ coincides with the "chromatic number" of $G$.
	\end{lemma}
	\begin{proof}
		We denote $"\mr(L_G)"$ the size of a minimal Rabin automaton recognising $"L_G"$ and $"\chi(G)"$ the "chromatic number" of $G$.
		\begin{description}
			\item[$\rlg \leq "\chi(G)"$:] Let $c: V \rightarrow [1,k]$ be a colouring of size $k$ of $G$. We will define a Muller automaton of size $k$ recognising $"L_G"$ and then use Proposition~\ref{prop: Rabin iff cycles} to show that we can put a "Rabin condition" "on top of" it. Let $\A_c=(Q,V, q_0,\dd, V, \F)$ be the Muller automaton defined by:
			\begin{itemize}
				\item $Q=\{1,2, \dots, k  \}$.
				\item $q_0 = 1$.
				\item $\dd(q,x)=(c(x), x)$ for $q\in Q$ and $x\in V$.
				\item A set $C\subseteq V$ belongs to $\F$ if and only if $C=\{v,u\}$ for two vertices $v,u\in V$ such that $(v,u)\in E$.
			\end{itemize} 
			
			The language recognised by $\A_c$ is clearly $"L_G"$, since the \kl(automata){output} produced by a word $w\in V^\oo$ is $w$ itself, and the acceptance condition $\F$ is exactly the \kl(condition){Muller condition} defining the language $"L_G"$. 
			
			Let $"G(\A_c)"=(Q, E_{\A_c})$ be the "graph associated to" $\A_c$.
			We will prove that the union of any pair of \kl(cycle){rejecting} "cycles" of $\A_c$ that have some "state in common" must be \kl(cycle){rejecting}. By Proposition \ref{prop: Rabin iff cycles} this implies that we can define a "Rabin condition" "on top of" $\A_c$.
			
			Let $\ell_1, \ell_2\subseteq E_{\A_c}$ be two "cycles" such that $"\gg(\ell_i)"\notin \F$ for $i\in \{1,2\}$ and such that $\mstates(\ell_1)\cap\mstates(\ell_2)\neq \emptyset$. We distinguish $3$ cases:
			\begin{itemize}
				\item $|\gg(\ell_i)|\geq 3$ for some $i\in \{1,2\}$. In this case, their union also has more than $3$ colours, so it must be rejecting.
				\item $\gg(\ell_i)=\{u,v\}$, $(u,v)\notin E$ for some $i\in \{1,2\}$. In that case, $\gg(\ell_1\cup \ell_2)$ also contains two vertices that are not connected by an edge, so it must be rejecting.
				\item $\gg(\ell_1)=\{v_1\}$ and $\gg(\ell_2)=\{v_2\}$. In this case, since from every state $q$ of $\A_c$ and every $v\in V$ $\dd(q,v)=(c(v),v)$,  the only state in each "cycle" is, respectively, $c(v_1)$ and $c(v_2)$. As $\ell_1$ and $\ell_2$ share some state, we deduce that $c(v_1)=c(v_2)$. If $v_1=v_2$, $\ell_1\cup \ell_2$ is \kl(cycle){rejecting} because $|\gg(\ell_1\cup \ell_2)|=1$. If $v_1\neq v_2$, it is also \kl(cycle){rejecting} because $c(v_1)=c(v_2)$ and therefore $(v_1, v_2)\notin E$. 
			\end{itemize}
			
			Since $\gg(\ell_i)$ is \kl(cycle){rejecting}, it does not consist on two vertices connected by some edge and we are always in some of the cases above. We conclude that we can put a "Rabin" condition "on top of" $\A_c$, obtaining a Rabin automaton recognising $"L_G"$ of \kl(automata){size} $k$.
			
			\item[$"\chi(G)" \leq \rlg$:] Let $\A=(Q,V, q_0, \dd, \GG, R)$ be a "Rabin" automaton of \kl(automata){size} $k$ recognising $"L_G"$ and $"G(\A)"=(Q, E_\A)$ "its graph". We will define a "colouring" of size $k$ of $G$, $c: V \rightarrow Q$.
			
			For each $v\in V$ we define a subset $Q_v\subseteq Q$ as:
			\[ Q_v = \{ q \in Q \; : \; \text{there is a "cycle" } \ell \text{ "containing" } q \text{ and } \gg(\ell)=\{v\} \}.\]
			
			For every $v\in V$, the set $Q_v$ is non-empty, as it must exist a (non-accepting) "run over" $v^\oo$ in $\A$. For each $v\in V$ we pick some $q_v \in Q_v$, and we define the "colouring" $c: V \rightarrow Q$ given by $c(v)=q_v$. 
			
			In order to prove that it is indeed a "colouring", we we will show that any two vertices $v,u\in V$ such that  $(v,u)\in E$ verify that $Q_v \cap Q_u = \emptyset$, and therefore they also verify $c(v) \neq c(u)$. Suppose by contradiction that there was some $q\in Q_v\cap Q_u$. Let us write $\ell_x$ for a "cycle" containing $q$ labelled with $x$, for $x\in \{v,u\}$. By the definition of $"L_G"$, both "cycles" $\ell_v$ and $\ell_u$ have to be \kl(cycle){rejecting} as $x^\oo \notin "L_G"$ for any $x\in V$. However, since $(u,v)\in E$, their union would be \kl(cycle){accepting}, contradicting Proposition~\ref{prop: Rabin iff cycles}. \qedhere		 
		\end{description}
	\end{proof}
	
	These results allow us to deduce the $\NP$-completeness of the minimisation of "Rabin" automata.
	\begin{theorem}\label{theorem: MinimisingRabinNP-comp}
		Given a "Rabin" automaton $\A$ and a positive integer $k$, deciding whether there is an equivalent "Rabin" automaton of \kl(automata){size} $k$ is $\NPc$. 
	\end{theorem}
	
	\begin{corollary}
		Given a "Streett" automaton $\A$ and a positive integer $k$, deciding whether there is an equivalent "Streett" automaton of \kl(automata){size} $k$ is $\NPc$.
	\end{corollary}

	\subsection{Parity and generalised Büchi automata recognising Muller languages can be minimised in polynomial time}\label{subsec: MinimisationParity}

	In Section \ref{subsec: MinimisationRabin} we have proven the $\NP$-hardness of the minimisation of "Rabin" automata showing a reduction that uses a language that is a \kl(language){Muller language}, that is, whether an infinite word $w$ belongs to the language only depends on $\minf(w)$.
	We may wonder whether \kl(language){Muller languages} could be used to prove $\NP$-hardness of minimising "parity" or "generalised Büchi" automata. We shall see now that this is not the case, as we prove in Propositions~\ref{prop : MinimisationParityPoly} and~\ref{prop : MinimisationGeneralisedBuchiPoly} that we can minimise parity and generalised Büchi automata recognising \kl(language){Muller languages} in polynomial time. Before proving these propositions we fix some notations.\\
	
	For the rest of the section, let $\SS$ stand for an input alphabet, $\F$ for a \kl(condition){Muller condition} over $\SS$ and let $\A=(Q,\SS, q_0, \dd, \GG, P)$ be a "parity" "automaton" recognising $L_\F$, with $\GG$ a finite subset of $\NN$.
	 Let $"G(\A)"=(V_\A, E_\A)$ be the "graph associated to" $\A$, and $"\iota"$ and $\gamma$ the projections over the input and output alphabet, respectively. We can suppose without loss of generality that $"G(\A)"$ is "strongly connected", since any "ergodic" "strongly connected component" of it must recognise the Muller language $"L_\F"$. We call a subgraph $\S=(V_S, E_S)$ of $G(\A)$ an \AP""$\A$-subgraph"" and we write $\AP""\mathit{Letters}""(\S)= \bigcup_{e\in E_S} "\iota"(e)$ and $\AP""\mathit{Max\hyphen Priority}""(\S)=\max\limits_{e\in E_S} \gg(e)$. 
	 
	 We say that $\S$ is \AP""complete"" if for all vertices $v\in V_S$ and letters $a\in \mletters(\S)$, $v$ has an outgoing edge $e$ in $\S$ such that $"\iota"(e)=a$. Let $C\subseteq \SS$ be a set of letters; we say that a subgraph $\S$ of $G(\A)$ is a \AP""$C$-Strongly Connected Component"" of $\A$ ($C$-SCC) if it is a "strongly connected" subgraph, "complete" and $\mletters(\S)=C$. 
	 	
	 \begin{fact}
	 	A "$C$-SCC" of $\A$ is naturally a "parity" "automaton" recognising the \kl(condition){Muller condition} $\F|_C = \F\cap \P(C)$ over $C$.
	 \end{fact}
 	\begin{lemma}\label{lemma: exists SCC}
 		For every subset of letters $C\subseteq \SS$ there is some "$C$-SCC" in $\A$.
 	\end{lemma} 
   \begin{proof}
  	It suffices to take an "ergodic" component of the graph obtained by restricting the transitions of $\A$ to those labelled with letters in $C$.
  \end{proof} 
	
	We denote $\AP""\mathtt{SCC\hyphen Decomposition}""(\S)$ a procedure that returns a list of the "strongly connected components" of an "$\A$-subgraph" $\S$ (what can be done in linear time in the number of vertices and edges \cite{Tarjan72DepthFirst},\cite{Sharir81StrongConnectivity}).
	
	We denote $\AP""\mathtt{Complete\hyphen SCC}""(C, \S)$ a procedure that takes as input a set of letters $C\subseteq \SS$ and an "$\A$-subgraph" that contains some $C$-SCC  and it returns in linear time a "$C$-SCC" of $\S$. We can do it simply by taking the restriction of $\S$ to edges labelled with letters in $C$, computing a "strongly connected component" decomposition of it  and choosing some "ergodic" component of it.
	
	The procedure $\AP""\mathtt{MaxInclusion}""(F)$ takes as input a family of subsets $F\subseteq \P(\SS)$ and returns a list containing the elements of $F$ that are maximal with respect to inclusion. This can be done in $\O(|\SS|^2 |F|^2)$.

	\begin{proposition}\label{prop : MinimisationParityPoly}
		Let $\F\subseteq \P(\SS)$ be a \kl(condition){Muller condition}. Given a "parity" automaton recognising $"L_\F"$, we can build in polynomial time a minimal parity automaton recognising $"L_\F"$.\footnote{In a previous version of this paper, we asked if this Proposition held in the state-based scenario. The answer is yes, since we can use the Zielonka tree of the Muller condition to build a minimal state-based parity automaton recognising the Muller condition. This result was independently proven by Philipp Meyer.}	
	\end{proposition}
	
	\begin{proof}
		By Proposition~\ref{prop: minimalParityZielonka}, we know that a minimal "parity" automaton recognising the \kl(language){Muller language} associated to the condition $\F\subseteq \P(\SS)$ can be constructed in linear time from its "Zielonka tree", $\ZF$. 
		 Let $\A$ be a parity automaton recognising $"L_\F"$. We give an algorithm (Algorithm \ref{algo: ZielonkaTree}) building $\ZF$ that works in $\O(d^2c^2n^4)$, where $d$ is the number of priorities in $\A$, $c=|\SS|$ and $n=|\A|$. This algorithm builds $\ZF$ recursively: first it calls the sub-procedure $\mathtt{AlternatingSets}$ (Algorithm \ref{algo: AlternatingComp}) that returns the labels of the children of the root of $\ZF$ and then it uses $"\mathtt{Complete \hyphen SCC}"$ to compute a "$C$-SCC" for each label $C$ of the children.
		

		\begin{algorithm}[ht]
			\caption{ $\mathtt{Zielonka \hyphen Tree(\S)}$}
			\label{algo: ZielonkaTree}
			\begin{algorithmic}[1]
				\State $\mathbf{\tinput:}$ An "$\A$-subgraph" $\S$
				\State $C_\S= \mletters(\S)$
				\State $\langle C_1, \dots, C_k\rangle = \mathtt{AlternatingSets}(\S)$
				\If{$k=0$ (Empty list)}  \label{lineZT : IfEmptyList}
				\State \Return $\ZF = \langle C_\S, \langle \emptyset \rangle \rangle$
				\Else 
				\For{$i=1 \dots k$}
				\State $\S_i = "\mathtt{Complete\hyphen SCC}"(C_i, \S)$ \label{lineZT : C-SCC}
				\State $T_i = \mathtt{Zielonka \hyphen Tree}(S_i)$
				\EndFor
				\State \Return $\ZF = \langle C_\S, \langle T_1,\dots ,T_k \rangle \rangle$
				\EndIf
			\end{algorithmic}
		\end{algorithm}
	
		\begin{algorithm}[ht]
		\caption{ The sub-procedure $\mathtt{AlternatingSets}$}
		\label{algo: AlternatingComp}
		\begin{algorithmic}[1]
			\State $\mathbf{\tinput:}$ An "$\A$-subgraph" $\S$ 
			\State $p=\maxpr(\S)$
			\State Compute $\S_p = (V_p, E_p)$, where $E_p=\{ e\in E \; : \; \gg(e)<p\}$ and $V_p$ are the vertices that have some outgoing edge labelled with a priority smaller than $p$.
			\State $L = \langle \S_1, \dots, \S_l \rangle = "\mathtt{SCC\hyphen Decomposition}"(\S_p)$ \label{line AS : SCC} 
			\State $\mathrm{AltSets}=\{ \}$ \Comment{Initialise an empty set}
			\For{$i=1,... l$} \label{line AS : For}
			\If{$\maxpr(\S_i)$ is odd if and only if $p$ is even}
			\State $\mathrm{AltSets}.\mathtt{add}(\mletters(\S_i))$
			\Else 
			\State $\mathrm{AltSets} =  \mathrm{AltSets} \, \cup \, \mathtt{AlternatingSets}(\S_i)$ \label{line AS : RecursiveAltenation}
			\EndIf\EndFor
		\State $\mathrm{MaxAltSets}="\mathtt{MaxInclusion}"(\mathrm{AltSets})$ \label{line AS : MaxLetters}
		\State \Return $\mathrm{MaxAltSets}$
		\end{algorithmic}
	\end{algorithm}

	\paragraph{Correctness of the algorithm.}
	 Let $T$ be the labelled tree returned by the algorithm $\mathtt{Zielonka \hyphen Tree}(G(\A))$. We prove that $T$ is the "Zielonka tree" of $\F$ by induction. We suppose without loss of generality that $\SS \in \F$, which implies that the maximal priority $p$ in $\A$ is even (as we have supposed that $\A$ is strongly connected).

			If $\ZF=\langle \SS, \langle \emptyset \rangle \rangle$ is a "leaf", we are going to prove that the procedure $\mathtt{AlternatingSets}(G(\A))$ returns an empty set, and therefore Algorithm \ref{algo: ZielonkaTree} enters in the conditional of line \ref{lineZT : IfEmptyList} and returns $\ZF$. Indeed, if $\ZF$ is a "leaf" it means that there is no subset of $\SS$ not belonging to $\F$. Therefore, if we remove recursively the transitions labelled with the maximal priority from $\A$ we must obtain a graph in which the maximal priority of any "strongly connected component" is even (if not, we would reject a word producing a run visiting all the transitions of this SCC). We conclude that the algorithm  $\mathtt{AlternatingSets}$ does not add any element to the set $\mathrm{AltSets}$ in the cycle of line \ref{line AS : For}.

			If $\ZF=\langle \SS, \langle \ZFi{1}, \dots, \ZFi{k} \rangle \rangle$, we prove that $\mathtt{AlternatingSets}(G(\A))$ returns a list of sets $C_1,\dots , C_k$ corresponding to the labels of the children of the root of $\ZF$. Then, we know that each "$C_i$-SCC" computed in line \ref{lineZT : C-SCC} corresponds to a parity automaton recognising the Muller condition associated to the subtree under the $i$-th child and thus, by induction hypothesis, $T = \ZF$.
			
			First, we remark that each set $C$ added to $\mathrm{AltSets}$ verifies that $C\notin \F$, since we only add it if it is the set of letters labelling one "strongly connected component" whose maximal priority is odd. As in line \ref{line AS : MaxLetters} we only keep the maximal subsets of $\mathrm{AltSets}$, it suffices to show that any maximal rejecting subset $C$ will be included in some set added to $\mathrm{AltSets}$. Let $C$ be such a set. By Lemma \ref{lemma: exists SCC} we know that there is some "$C$-SCC" $\S_C$ in $\A$, and it must verify that the maximal priority on it is odd. Let $\S'$ be a "strongly connected" subgraph of $"G(\A)"$ containing $\S_C$. Since $C$ is maximal amongst the rejecting subsets, either $\mletters(\S')=C$ or the maximal priority of $\S'$ is even. Therefore, in line \ref{line AS : RecursiveAltenation} we will disregard any SCC $\S'$ containing $\S_C$ that does not verify $\mletters(\S')=C$ and we will recursively inspect it until finding one labelled with $C$. We note, that all subgraphs considered in this process contain $\S_C$ since the maximal priority of it is odd.

		\paragraph{Complexity analysis.} 
		Let $n=|\A|$, $d$ be the number of priorities appearing in $\A$, $c=|\SS|$, $f(n,d,c)$ be the complexity of $\mathtt{Zielonka \hyphen Tree}(G(\A))$ and $g(n,d,c)$ the complexity of $\mathtt{AlternatingSets}(G(\A))$. We start by obtaining an upper bound for $g(n,d,c)$. It is verified:
		\[ g(n,d,c) \leq g(n_1,d-2,c) + \dots + g(n_l,d-2,c) + \O(c^2n^2) \]
		for some $1 \leq n_1,\dots, n_l\leq n$ such that $\sum_{i=1}^{l}n_i \leq n$.
		
		We obtain a weighted-tree such that the sum of the weights of the children of some node $\tt$ is at most the weight of $\tt$. The height of this tree is at most $\lceil d/2 \rceil$ and the weight of the root is $n$. One such tree has less than $\O(dn)$ nodes, so $g(n,d) = \O(dc^2n^3)$.
		
		For the complexity of $f(n,d,c)$ we have a similar recurrence:
		\[ f(n,d,c) \leq f(n_1, d-1,c) + \dots + f(n_k,d-1,c) + k\O(cn)+ g(n,d,c) \]
		for some $1 \leq n_1,\dots, n_k\leq n-1$ such that $\sum_{i=1}^{k}n_i \leq n$.

		Replacing $g(n,d,c)$ by $\O(dc^2n^3)$ and using an argument similar to the above one, we obtain that $f(n,d,c)=\O(d^2c^2n^4)$.
		\qedhere		 
	\end{proof}


\begin{proposition}\label{prop : MinimisationGeneralisedBuchiPoly}	
	Let $\F\subseteq \P(\SS)$ be a \kl(condition){Muller condition}. If $"L_\F"$ can be recognised by a "generalised Büchi" (resp. "generalised co-Büchi") automaton, then, it can be recognised by one such automaton with just one state. Moreover, this minimal automaton can be built in polynomial time from any "generalised Büchi" (resp. generalised co-Büchi) automaton recognising $"L_\F"$.
\end{proposition}
\begin{proof}
	We present the proof for "generalised Büchi" automata, the generalised co-Büchi case being dual. The language $"L_\F"$ can be recognised by a deterministic "generalised Büchi" automaton if and only if it can be recognised by a Büchi automaton, and this is possible if and only if the Zielonka tree $\ZF$ has at most height $2$ and in the case it has exactly height $2$, the set $\SS$ is accepting \cite{CCF21Optimal}.
	
	Suppose that  $\F$ verifies this condition, and let its "Zielonka tree" be $\ZF=\ab \langle \SS, \langle \langle A_1, \langle \emptyset \rangle \rangle ,\langle A_2, \langle \emptyset \rangle \rangle, \dots, \langle A_k, \langle \emptyset \rangle \rangle \rangle \rangle$, verifying $\SS\in \F$ and $A_i\notin \F$ for $i\in \{1,\dots, k\}$. Then, it is easy to check that the following "generalised Büchi" automaton $\A_{\mathit{min}}=(\{q_0\},\SS, q_0, \dd, \SS, \{B_1,\dots, B_k\})$ of size $1$ recognises $"L_\F"$:
	\begin{itemize}
		\item $Q=\{q_0\}$.
		\item $\dd(q,c)=(q,c)$, for $c \in \SS$.
		\item $B_i = \SS \setminus A_i$ are the sets defining the "generalised Büchi" condition.
	\end{itemize}
	
	Finally, we prove that we can build this automaton from a given "generalised Büchi" automaton recognising $"L_\F"$ in polynomial time. We remark that in order to build $\A_{\mathit{min}}$ it suffices to identify the maximal rejecting subsets $A_1, A_2, \dots, A_k\subseteq \SS$.
	Let  $\A'=(Q,\SS, q_0, \dd, \GG, \{B_1',\dots, B_r'\})$ be a "generalised Büchi" automaton recognising $"L_\F"$, let $"G(\A')"=(V_{\A'}, E_{\A'})$ be the "graph associated to" $\A'$, and let $"\iota"$ and $\gamma$ be the projections over the input and output alphabet, respectively.  For each subset $B_i'$, $i\in \{1,\dots, r\}$ we consider the subgraph $G_i=(V_i, E_i)$, obtained as the restriction of $"G(\A')"$ to the transitions that do not produce a letter in $B_i'$, that is, $V_i=V_{\A'}$, $E_i=\{e\in E_{\A'} \: : \: \gg(e)\notin B_i'\}$. We compute the "strongly connected component" of $G_i$, $\S_{i,1},\dots, \S_{i,l_i}$, and for each of them we consider the projection over the input alphabet, $C_{i,j}="\iota"(\S_{i,j})\subseteq \SS$. In this way, we obtain a colection of subsets of $\SS$, $\{C_{i,j}\}_{\substack{1\leq i\leq k; 1 \leq j \leq l_i}}$, and all of them have to be rejecting. We claim that the maximal subsets amongst them are the maximal rejecting subsets of $\SS$. Indeed, if $A\subseteq \SS$ is a rejecting subset, a run over $\A'$ reading all the letters of $A$ eventually does not see any transition producing a colour in $B_i'$, for some $i$, so this run is contained in a strongly connected component of $G_i$, and $A$ is therefore a subset of $C_{i,j}$, for some $j$.

\end{proof}

	\section{Memory in games}\label{section : MemoryDefinitions}
	
	In this section, we introduce the definitions of "games", "memories" and "chromatic memories" for games, as well as the distinction between games where we allow $\ee$-transitions and "$\ee$-free" games. We show in Section~\ref{subset: epsilon-transitions} that the "memory requirements" for this two latter classes of games might differ.\\

	\subsection{Games}
	A \AP""game"" is a tuple $\G=(V=V_E\uplus V_A, E, v_0, \gg:E\rightarrow \GG\cup \{\varepsilon\}, \macc)$ where $(V,E)$ is a "directed graph" together with a partition of the vertices $V=V_E\uplus V_A$, $v_0$ is an initial vertex, $\gg$ is a colouring of the edges and $\macc$ is a winning condition defining a subset $\WW \subseteq \GG^\oo$. The letter $\varepsilon$ is a neutral letter, and we impose that there is no cycle in $\G$ labelled exclusively with $\varepsilon$.	
	 We say that vertices in $V_E$ belong to \emph{Eve} (also called the \emph{existential player}) and those in $V_A$ to \emph{Adam} (\emph{universal player}). We suppose that each vertex in $V$ has at least one outgoing edge. A game that uses a winning condition of type $X$ (as defined in Section~\ref{subsection : AutomataDefinitions}) is called an $X$-game.

	A \AP""play"" in $\G$ is an infinite path $\rr\in E^\oo$ produced by moving a token along edges starting in $v_0$: the player controlling the current vertex chooses what transition to take. Such a play produces a word $\gg(\rr)\in (\GG\cup \{\varepsilon\})^\oo$. Since no cycle in $\G$ consists exclusively of $\varepsilon$-colours, after removing the occurrences of $\varepsilon$  from $\gg(\rr)$ we obtain a word in $\GG^\oo$, that we will call the \intro(game){output} of the play and we will also denote $\gg(\rr)$ whenever no confusion arises. The play is \AP""winning"" for Eve if the output belongs to the set $\WW$ defined by the acceptance condition, and winning for  Adam otherwise. A \AP""strategy"" for Eve in $\G$ is a function prescribing how Eve should play. Formally, it is a function $\ss: E^* \rightarrow E$ that associates to each partial play ending in a vertex $v\in V_E$ some outgoing edge from $v$. A play $\rr\in E^\oo$ \AP""adheres"" to the strategy $\ss$ if for each partial play $\rr'\in E^*$ that is a prefix of $\rr$ and ends in some state of Eve, the next edge played coincides with $\ss(\rr')$.
	We say that Eve \AP""wins"" the game $\G$ if there is some strategy $\ss$ for her such that any "play" that adheres to $\ss$ produces a winning play for her (in this case we say that $\ss$ is a \emph{winning strategy)}. 
	
	We will also study games without $\varepsilon$-transitions. We say that a game $\G$ is \AP""$\varepsilon$-free"" if $\gg(e)\neq \varepsilon$ for all edges $e\in E$.
	
	
	\subsection{Memory structures}
	
	We give the definitions of the following notions from the point of view of the existential player, Eve. Symmetric definitions can be given for the universal player (Adam), and all results of Section~\ref{section : ChromaticMemory} can be dualised to apply to the universal player.
	
	A \AP""memory structure for the game $\G$"" is a tuple $\M_\G=(M, m_0, \mu)$ where $M$ is a set of states, $m_0\in M$ is an initial state and $\mu: M \times E \rightarrow M$ is an update function (where $E$ denotes the set of edges of the game).  We extend the function $\mu$ to $M\times E^*$ in the natural way. 
	 We can use such a memory structure to define a "strategy" for Eve using a function $\nextmove: V_E \times M \rightarrow E$, verifying that $\nextmove(v,m)$ is an outgoing edge from $v$. After each move of a "play" on $\G$, the state of the memory $\M_\G$ is updated using $\mu$; and when a partial play arrives to a vertex $v$ controlled by Eve she plays the edge indicated by the function $\nextmove(v,m)$, where $m$ is the current state of the memory. We say that the memory structure $\M_\G$ \AP""sets a winning strategy"" in $\G$ if there exists such a function $\nextmove$ defining a "winning strategy" for Eve.

	
	We say that $\M_\G$ is a \AP""chromatic memory"" if there is some function $\mu': M \times \GG \rightarrow M$ such that $\mu(m,e) = \mu'(m, \gg(e))$  for every edge $e\in E$ such that $\gg(e)\neq \varepsilon$, and $\mu(m,e)=m$ if $\gg(e)=\varepsilon$. That is, the update function of $\M_\G$ only depends on the colours of the edges of the game.
	
	
	Given a winning condition $\WW\subseteq \GG^\oo$, we say that $\M=(M, m_0, \mu)$ is an \AP""arena-independent memory"" for $\WW$ if for any $\WW$-game $\G$ won by Eve, there exists some function $\nextmove_\G: V_E \times M \rightarrow E$ "setting a winning strategy" in $\G$. We remark that such a memory is always chromatic.
	
	The \intro(memory){size} of a memory structure is its number of states.

	
	Given a \kl(condition){Muller condition} $\F$, we write $\AP""\mathfrak{mem}_{\mathit{gen}}(\F)""$ (resp. $\AP""\mathfrak{mem}_{\mathit{chrom}}(\F)""$) for the least number $n$ such that for any $\F$-game that is won by Eve, she can win it using a memory (resp. a chromatic memory) of size $n$. We call $\mathfrak{mem}_{\mathit{gen}}(\F)$ (resp. $\mathfrak{mem}_{\mathit{chrom}}(\F)$) the \emph{general memory requirements (resp. chromatic memory requirements) of $\F$}. We write $\AP""\mathfrak{mem}_{\mathit{ind}}(\F)""$ for the least number $n$ such that there exists an "arena-independent memory" for $\F$ of size $n$.
	
	We define respectively all these notions for "$\varepsilon$-free" $\F$-games. We write $\AP""\mathfrak{mem}_{\mathit{gen}}^{\ee\hyphen \mathrm{free}}(\F)""$, $\AP""\mathfrak{mem}_{\mathit{chrom}}^{\ee\hyphen \mathrm{free}}(\F)""$ and $\AP""\mathfrak{mem}_{\mathit{ind}}^{\ee\hyphen \mathrm{free}}(\F)""$ to denote, respectively, the minimal general memory requirements, minimal chromatic memory requirements and minimal size of an arena-independent memory for $\varepsilon$-free $\F$-games.
	
	\begin{remark}
		We remark that these quantities verify that $\mF \leq \cmF \leq  \indmF$ and that $\mathfrak{mem}_{\mathit{X}}^{\ee\hyphen \mathrm{free}}(\F)\leq \mathfrak{mem}_{\mathit{X}}(\F)$ for $X\in \{ \mathit{gen},  \mathit{chrom},  \mathit{ind}\}$.
	\end{remark}

	A family of games is \AP""half-positionally determined"" if  for every game in the family that is won by Eve, she can win using a "strategy" given by a "memory structure" of size $1$.
	
	\begin{lemma}[\cite{Klarlund94Determinacy, Zielonka1998infinite}]\label{lemma: RabinPositional}
		"Rabin"-games are half-positionally determined.
	\end{lemma}
	
	If $\A$ is a "Rabin" "automaton" recognising a \kl(condition){Muller condition} $\F$, given an $\F$-game $\G$ we can perform a standard product construction $\G \ltimes \A$ to obtain an equivalent game using a Rabin condition that is therefore half-positionally determined. This allows us to use the automaton $\A$ as an "arena-independent memory" for $\F$.
	
	\begin{lemma}[Folklore]\label{lemma: RabinAutomataAsMemory}
		Let $\F$ be a \kl(condition){Muller condition}. We can use a "Rabin" automaton $\A$ recognising $"L_\F"$ as an "arena-independent memory" for $\F$.
	\end{lemma}
	
	\subsection{The general memory requirements of Muller conditions}
	
	The "Zielonka tree" (see Definition~\ref{def: ZielonkaTree}) was introduced by Zielonka in \cite{Zielonka1998infinite}, and in \cite{DJW1997memory} it was used to characterise the general memory requirements of Muller games as we show next.

	\begin{definition}
		Let $\F$ be a Muller condition and $\ZF=\langle \GG, \langle \ZFi{1}, \dots \ZFi{k}\rangle\rangle$ its Zielonka tree. We define the number $\AP""\mathfrak{m}_{\ZF}""$ recursively as follows:
		\begin{equation*}
		\mathfrak{m}_{\ZF} =	\begin{cases}
				1 & \text{ if } \ZF \text{ is a "leaf"}\\[1mm]
				\max \{ \mathfrak{m}_{\ZFi{1}}, \dots , \mathfrak{m}_{\ZFi{k}}\} & \text{ if } \GG\notin  \F \text{ and }  \ZF \text{ is not a "leaf" }\\[1mm]
				\sum\limits_{i=1}^k \mathfrak{m}_{\ZFi{i}} & \text{ if } \GG\in  \F \text{ and }  \ZF \text{ is not a "leaf" }.
			\end{cases}
		\end{equation*}
	\end{definition}

	\begin{proposition}[\cite{DJW1997memory}]\label{prop: optimalMemoryDJW}
		For every  Muller condition $\F$, $\mF = \mathfrak{m}_{\ZF}$. That is,
		\begin{enumerate}
			\item If Eve wins an $\F$-game, she can win it using a strategy given by a (general) "memory structure" of size at most $\mathfrak{m}_{\ZF}$. 
			\item There exists an $\F$-game (with $\ee$-transitions) won by Eve such that she cannot win it using a strategy given by a memory structure of size strictly smaller than $\mathfrak{m}_{\ZF}$.
		\end{enumerate}
	\end{proposition}

\subsection{Memory requirements of \texorpdfstring{$\varepsilon$}{e}-free games}\label{subset: epsilon-transitions}
In \cite{Zielonka1998infinite} and \cite{Kopczynski2006Half} it was noticed that there can be major differences regarding the memory requirements of winning conditions depending on the way the games are coloured. We can differentiate 4 classes of games, according to whether we colour vertices or edges, and whether we allow or not the neutral colour $\varepsilon$:
\begin{enumerate}[A)]
	\item State-coloured $\varepsilon$-free games.
	\item General state-coloured games.
	\item Transition-coloured "$\varepsilon$-free" games.
	\item General transition-coloured games.
\end{enumerate}
(In the previous sections we have only defined transition-coloured games). We remark that the "memory requirements" of any condition are the same for general state-coloured games and general transition-coloured games.

In \cite{Zielonka1998infinite}, Zielonka showed that there are Muller conditions that are half-positional over state-coloured $\ee$-free games, but they are not half-positional over general state-coloured games, and he exactly characterises half-positional Muller conditions in both cases.

However, when considering transition-coloured games, this ``bad behaviour'' does not appear: in both general games and "$\varepsilon$-free games", "half-positional" \kl(condition){Muller conditions} correspond exactly to "Rabin conditions" (Lemma~\ref{lemma : positionalRabin}). In particular, the characterisation of Zielonka of "half-positional" Muller conditions for state-coloured $\ee$-free games does not generalise to transition-coloured "$\varepsilon$-free" games.

Nevertheless, the matching upper bounds for the memory requirements of Muller conditions appearing in \cite{DJW1997memory} are given by transition-labelled games using $\varepsilon$-transitions. An interesting question is whether we can produce upper-bound examples using "$\varepsilon$-free" games.
In this section we answer this question negatively. We show in Proposition~\ref{prop : epsilon-more-memory} that, for every $n\geq 2$, there is a Muller condition $\F$ such that Eve can use "memories" of \kl(memory){size} $2$ to win "$\varepsilon$-free" $\F$-games where she can force a win, but that she might need $n$ memory states to win $\F$-games with $\varepsilon$ transitions. That is, $\mFfree < \mF$ and the gap can be arbitrarily large. In Section~\ref{subsection : ChromaticMemory-RabinAutomata} we will see that this is not the case for "chromatic memories": $\cmF=\cmFfree$ for any \kl(condition){Muller condition} $\F$.

\begin{lemma}\label{lemma : positionalRabin}
	For any Muller condition $\F\subseteq \P(\GG)$,  $\F$ is "half-positional" determined over transition-coloured "$\varepsilon$-free games" if and only if $\F$ is half-positional determined over general transition-coloured games. That is, $\mF=1$ if and only if $\mFfree = 1$.
\end{lemma}
\begin{proof}
	Since $\mFfree \leq \mF$ for all conditions, it suffices to see that if $\F$ is a Muller condition such that $1 < \mF$, then we can find an $\varepsilon$-free $\F$-game won by Eve where she cannot win positionally.
	
	By the characterisation of \cite{DJW1997memory} (\Cref{prop: optimalMemoryDJW}), we know that $\mF=1$ if and only if every node of the "Zielonka tree" of $\F$ labelled with an accepting set of colours has at most one child. Therefore, if $\F\subseteq \P(\GG)$ is a Muller condition verifying $1 < \mF$, we can find a node in $\ZF$ labelled with $A\in \F$ and with two different children, labelled respectively with $B_1, B_2\subseteq \GG$, verifying that $B_i \notin \F$ for $i\in \{1,2\}$ and such that $B_i \nsubseteq B_j$, for $i\neq j$. We consider the "$\varepsilon$-free game" consisting in one initial state $v_0$ controlled by Eve with two outgoing edges, one leading to a cycle that comes back to $v_0$ after producing the colours in $B_1$, and the other one leading to a cycle producing the colours in $B_2$. It is clear that Eve wins this game, but she needs at least $2$ memory states to do so.
\end{proof}

\begin{proposition}\label{prop : epsilon-more-memory}
	For any integer $n\geq 2$, there is a set of colours $\GG_n$ and a \kl(condition){Muller condition} $\F_n \subseteq \P(\GG_n)$ such that $\mFfree=2$ and $\mF=n$.
\end{proposition}

\begin{proof}
	Let us consider the set of colours $\GG_n=\{1, \dots, n\}$ and the Muller condition \[\F_n=\{ A \subseteq \GG_n \: : \: |A| > 1\}.\]
	The "Zielonka tree" of $\F_n$ is pictured in Figure~\ref{Fig : ZielonkaTreeSimple}, where round nodes represent nodes whose label is an accepting set, and rectangular ones, nodes whose label is a rejecting set.
	
	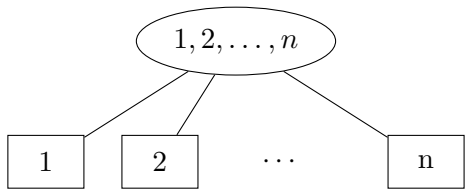
\begin{figure}[ht]
		\centering
		\begin{tikzpicture}[square/.style={regular polygon,regular polygon sides=4}, align=center,node distance=2cm,inner sep=3pt]
		
		\node at (0,2.6) [draw, ellipse, minimum height=0.9cm, minimum width = 1.5cm] (R) {$1,2, \dots, n$};
		
		\node at (-2.5,1) [draw, rectangle ,minimum height=0.7cm, minimum width = 1cm] (0) {1};
		\node at (-1,1) [draw, rectangle, ,minimum height=0.7cm, minimum width = 1cm] (1) {2};
		\node at (0.6,1) (dots) {$\cdots$};
		\node at (2.5,1) [draw, rectangle, ,minimum height=0.7cm, minimum width = 1cm] (2) {n};
		
		\draw   
		(R) edge (0)
		(R) edge (1)
		(R) edge (2);
		
		\end{tikzpicture}
		\caption{Zielonka tree $\ZFi{n}$.}
		\label{Fig : ZielonkaTreeSimple}
	\end{figure}
	
	By the characterisation of \cite{DJW1997memory} (\Cref{prop: optimalMemoryDJW}), we know that $\mF=n$, that is, there is a game with $\varepsilon$-transitions won by Eve where she needs at least $n$ memory states to force a win. We are going to prove that if Eve wins an "$\varepsilon$-free" $\F$-game, then she can win it using only $2$ memory states.
	
	Let $\G=(V= V_E \uplus V_A, E, v_0, \gg, \F)$ be a an $\varepsilon$-free $\F$-game. First, we can suppose that Eve wins the game where we change the initial vertex $v_0$ to any other vertex of $V$. Indeed, since $\F$ is a prefix-independent condition, we can restrict ourselves to the winning region of Eve (and any strategy has to ensure that she remains there). From any vertex $v\in V_E$, there is one outgoing transition coloured with some colour in $\GG$, since the game is $\varepsilon$-free. We associate to each vertex $v\in V_E$ one such colour, via a mapping $c:V_E \rightarrow \GG_n$, obtaining a partition $V_E= V_1 \uplus \dots \uplus V_n$ verifying that $v\in V_x$ implies that there is some transition labelled with $x$ leaving $v$, for $x\in \GG_n$. We denote $\ss_0:V_E\rightarrow E$ one application that maps each vertex $v\in V_E$ to one outgoing edge labelled with $c(v)$.
	
	Moreover, for any $v\in V$, since Eve can win from $v$, there is some strategy in $\G$ forcing to see some colour $y\in \GG_n$, $y\neq c(v)$. For each colour $x\in\GG_n$, we consider the game $\G_x= (V= V_E \uplus V_A, E, v, \gg, \mathit{Reach}(\GG_n\setminus \{x\}))$, where the underlying graph is the same as in $\G$, $v\in V$ is an arbitrarily vertex and the winning condition consists in reaching some colour different from $x$. Eve can win this game starting from any vertex $v\in V$, and we can fix a positional winning strategy for her that does not depend on the initial vertex (since reachability games are uniformly positional determined \cite{GradelThomasWilke2002AutLogGames}). We denote $\ss_x:V_E\rightarrow E$ the choice of edges defining this strategy. Moreover, we can pick $\ss_x$ such that it coincides with $\ss_0$ outside $V_x$. 
	
	We define a "memory structure" $\M_\G= (M=\{m_0,m_1\}, m_0, \mu)$ for $\G$ and a $\nextmove$ function describing the following strategy: the state $m_0$ will be used to remember that we have to see the colour $x$ corresponding to the component $V_x$ that we are in. As soon as we arrive to a vertex controlled by Eve, we use the next transition to accomplish this and we can change to state $m_1$ in $\M_\G$. The state $m_1$ will serve to follow the positional strategy $\ss_x$ reaching one colour different from $x$. We will change to state $m_0$ if we arrive to some state in $V_E$ not in $V_x$ (this will ensure that we will see one colour different from $x$), or if Eve produces a colour different from $x$ staying in the component $V_x$. If she does not produce this colour and we do not go to a vertex in $V_E \setminus V_x$, that means that (since $\ss_x$ ensures that we will see a colour different from $x$), Adam will take some transition coloured with some colour different from $x$. 
	This "memory structure" is formally defined as follows:
	
	\begin{itemize}
		\item The set of states is $M=\{m_0, m_1\}$, being $m_0$ the initial state.
		\item The update function $\mu: M\times E \rightarrow M$ is defined as:
		\[\begin{cases}
		\mu(m_i, (v,v'))=m_i & \text{ if } v\in V_A, \text{ for } i\in \{1,2\},\\
		\mu(m_0, (v,v'))=m_1 & \text{ if } v\in V_E,\\
		\mu(m_1, (v,v'))=m_0 & \text{ if } v,v'\in V_E \text{ and } c(v)\neq c(v'),\\
		\mu(m_1, (v,v'))=m_0 & \text{ if } v\in V_E \text{ and } c(v)\neq \gg((v,v')),\\
		\mu(m_1, (v,v'))=m_1 & \text{ in any other case.}
		\end{cases} \]
		\item The function $\nextmove: M\times V_E \rightarrow E$ is defined as:
		\[\begin{cases}
		\nextmove(m_0, v)=\ss_0(v),\\
		\nextmove(m_1, v)=\ss_{c(v)}(v).
		\end{cases}\qedhere	 \]	
	\end{itemize} 
\end{proof}

\begin{remark}
	The condition of the previous proof also provides an example of a condition that is half-positional over "$\ee$-free" state-coloured arenas, but for which we might need memory $n$ in general state-coloured arenas (similar examples can be found in \cite{Zielonka1998infinite, Kopczynski2006Half}).
\end{remark}
However, the question raised in \cite{Kopczynski2006Half} of whether there can be conditions (that cannot be Muller ones) that are "half-positional" only over "$\ee$-free" games remains open.

	\section{The chromatic memory requirements of Muller conditions}\label{section : ChromaticMemory}
	In this section we present the main contributions concerning the "chromatic memory requirements" of \kl(condition){Muller conditions}. In Section~\ref{subsection : ChromaticMemory-RabinAutomata}, we prove that the "chromatic memory requirements" of a Muller condition (even for "$\ee$-free" games) coincide with the size of a minimal Rabin automaton recognising the Muller condition (Theorem~\ref{theorem : EquivRabin-ChromaticMemory}). In Section~\ref{subsection : complexity ChromaticMemory} we deduce that determining the "chromatic memory requirements" of a \kl(condition){Muller condition} is $\NPc$, for different representations of the condition. Finally, this results allow us to answer in Section~\ref{subset: KopczynskiConjecture} the question appearing in \cite{Kopczynski2006Half,Kopczynski2008PhD} of whether the "chromatic memory requirements" coincide with the "general memory requirements" of winning conditions.

\subsection{Chromatic memory and Rabin automata}\label{subsection : ChromaticMemory-RabinAutomata}
	In this section we prove Theorem~\ref{theorem : EquivRabin-ChromaticMemory}, establishing the equivalence between the "chromatic memory requirements" of a \kl(condition){Muller condition} (also for "$\ee$-free" games) and the size of a minimal "Rabin" automaton recognising the Muller condition.
	
	Lemma~\ref{lemma : Chromatic=Independent} appears in Kopczyński’s PhD thesis~\cite[Proposition 8.9]{Kopczynski2008PhD} (unpublished). We present a similar proof here.
	
	\begin{lemma}[\cite{Kopczynski2008PhD}]\label{lemma : Chromatic=Independent}
		Let $\F$ be a Muller condition. Then, $\cmF = \indmF.$ That is, there is an $\F$-game $\G$ won by Eve such that any "chromatic memory" for $\G$ "setting a winning strategy" has \kl(memory){size} at least $\indmF$, where $\indmF$ is the minimal size of an "arena-independent memory" for $\F$.
		
		The same result holds for "$\ee$-free" games: $\cmFfree=\indmFfree$.
	\end{lemma}

	\begin{proof}
		We present the proof for $\cmF = \indmF$, the proof for the "$\varepsilon$-free" case being identical, since we do not add any $\varepsilon$-transition to the games we consider.

		It is clear that $\cmF \leq \indmF$, since any "arena-independent memory" for $\F$-games has to be "chromatic". We will prove that it is not the case that $\cmF < \indmF$. 
		Let $\M_1, \cdots , \M_n$ be an enumeration of all "chromatic memory structures" of size strictly less than $\indmF$. 
		By definition of $\indmF$, for any of the memories $\M_j$ there is some $\F$-game $\G_j=(V_j, E_j, v_{0_j}, \gg_j)$ won by Eve such that no function $\nextmove_{\G_j}: M_j\times V_{j}\rightarrow E_j$ "setting a winning strategy" in $\G_j$ exists. 
		We define the disjoint union of these games, $\G= \biguplus\limits_{i=1}^{n} \G_i$, as the game with an initial vertex $v_0$ controlled by Adam, from which he can choose to go to the initial vertex of any of the games $\G_i$ producing the letter $a\in \GG$ (for some $a\in \GG$ fixed arbitrarily), and such the rest of vertices and transitions of $\G$ is just the disjoint union of those of the games $\G_i$.  Eve can win this game, since no matter the choice of Adam we arrive to some game where she can win. However, we show that she cannot win using a "chromatic memory" strictly smaller than $\indmF$. Suppose by contradiction that she wins using a chromatic memory $\M=(M,m_0, \mu)$, $|\M|< \indmF$. We let $m_0' = \mu(m_0,a)$, and we consider the memory structure $\M' = (M, m_0' , \mu)$. Since $|\M'|< \indmF$, $\M' = \M_i$ for some $i\in \{1, \dots, n\}$, and therefore Adam can choose to take the transition leading to $\G_i$, where Eve cannot win using this memory structure. This contradicts the fact that Eve wins $\G$ using $\M$.
	\end{proof}

	\begin{theorem}\label{theorem : EquivRabin-ChromaticMemory}
		Let $\F \subseteq \P(\GG)$ be a Muller condition. The following quantities coincide:
		\begin{enumerate}
			\item The \kl(automata){size} of a minimal deterministic "Rabin" automaton recognising $"L_\F"$, $"\mr(L_\F)"$.
			\item The \kl(memory){size} of a minimal "arena-independent memory" for $\F$, $\indmF$.
			\item The \kl(memory){size} of a minimal "arena-independent memory" for "$\varepsilon$-free" $\F$-games, $\indmFfree$.
			\item The "chromatic memory requirements" of $\F$, $\cmF$.
			\item The "chromatic memory requirements" of $\F$ for "$\varepsilon$-free" games, $\cmFfree$.
		\end{enumerate}
	\end{theorem}
	
	\begin{proof}
		The previous Lemma~\ref{lemma : Chromatic=Independent}, together with Lemma~\ref{lemma: RabinAutomataAsMemory}, prove that
		\[ \indmFfree=\cmFfree \leq \cmF=\indmF \leq "\mr(L_\F)".\] 
		In order to prove that $"\mr(L_\F)"\leq \indmFfree$, we are going to show that we can put a Rabin condition on top of any "arena-independent memory" for $\varepsilon$-free $\F$-games $\M$, obtaining a Rabin automaton recognising $"L_\F"$ and having the same size than $\M$.
		
		Let $\M = (M, m_0, \mu:M\times \GG \rightarrow M)$ be an "arena-independent memory" for "$\varepsilon$-free" $\F$-games. First, we remark that we can suppose that every state of $\M$ is accessible from $m_0$ by some sequence of transitions. 
		 We define a Muller "automaton" $\A_\M$ using the underlying structure of $\M$: $\A_\M= (M, \GG, m_0, \dd , \GG, \F)$, where the transition function $\dd$ is defined as $\dd(m,a)=(\mu(m,a),a)$, for $a\in \GG$. Since the \kl(automata){output} produced by any word $w\in \GG^\oo$ is $w$ itself and the accepting condition is $\F$, this automaton trivially accepts the language $"L_\F"$. We are going to show that the Muller automaton $\A_\M$ satisfies the second property in Proposition \ref{prop: Rabin iff cycles}, that is, that for any pair of "cycles" in $\A_\M$ with some "state in common", if both are \kl(cycle){rejecting} then their union is also \kl(cycle){rejecting}. This will prove that we can put a "Rabin" condition "on top of" $\A_\M$.
		
		 Let $\ell_1$ and $\ell_2$ be two rejecting "cycles" in $\A_\M$  such that $m\in M$ is contained in both $\ell_1$ and $\ell_2$. We suppose by contradiction that their union $\ell_1 \cup \ell_2$ is an \kl(cycle){accepting} "cycle". We will build an "$\varepsilon$-free" $\F$-game that is won by Eve, but where she cannot win using the "memory" $\M$, leading to a contradiction. Let $a_0a_1\dots a_k\in \GG^*$ be a word labelling a path to $m$ from $m_0$ in $\M$, that is, $\mu(m_0, a_0\dots a_k)= m$. We define the $\varepsilon$-free $\F$-game $\G= (V=V_E, E, v_0, \gg:E\rightarrow \GG, \F)$ 
		 as the game where there is a sequence of transitions labelled with $a_0\dots a_k$ from $v_0$ to one vertex $v_m$ controlled by Eve (the only vertex in the game where some player has to make a choice). From $v_m$, Eve can choose to see all the transitions of $\ell_1$ before coming back to $m$ (producing the corresponding colours), or to see all the transitions of $\ell_2$ before coming back to $m$.
		 
		 First, we notice that Eve can win the game $\G$: since $\ell_1\cup \ell_2$ is accepting, she only has to alternate between the two choices in the state $v_m$. However, there is no function $\nextmove : M\times V_E \rightarrow E$ setting up a winning strategy for Eve. Indeed, for every partial play ending in $v_m$ and labelled with $a_0a_1\dots a_s$, it is clear that $\mu(m_0,a_0\dots a_s)=m$ (the memory is at state $m$). If $\nextmove(m,v_m)$ is the edge leading to the cycle corresponding to $\ell_1$, no matter the value $\nextmove$ takes at the other pairs, all plays will stay in $\ell_1$, so the set of colours produced infinitely often would be $\gg(\ell_1)$ which is loosing for Eve. The result is symmetric if $\nextmove(m,v_m)$ is the edge leading to the other cycle. We conclude that $\M$ cannot be used as a memory structure for $\G$, a contradiction.
	\end{proof}

	\subsection{The complexity of determining the chromatic memory requirements of a Muller condition}\label{subsection : complexity ChromaticMemory}
	As shown in \cite{DJW1997memory}, the "Zielonka tree" of a \kl(condition){Muller condition} directly gives its "general memory requirements". In this section, we see that it follows from the previous results that determining the "chromatic memory requirements" of a Muller condition is $\NPc$, even if it is represented by its Zielonka tree. 
	
	\begin{proposition}
		Given the "Zielonka tree" $\ZF$ of a \kl(condition){Muller condition} $\F$, we can compute in $\O(|\ZF|)$ the "memory requirements" for $\F$-games, $\mF$.
	\end{proposition}
	Using Algorithm~\ref{algo: ZielonkaTree} we can compute the "Zielonka tree" of a Muller condition from a parity automaton recognising that condition. We deduce the following result.

	\begin{corollary}
		Given a parity automaton $\P$ recognising a \kl(condition){Muller condition} $\F$, we can compute in polynomial time in $|\P|$ the "memory requirements" for $\F$-games, $\mF$.
	\end{corollary}

	\begin{theorem}\label{theorem : ChromaticMemoryNP-hard}
	Given a positive integer $k>0$ and a Muller condition $\F$ represented as either: 
		\begin{enumerate}[a)]
			\item The Zielonka tree $\ZF$.
			\item A parity automaton recognising $"L_\F"$.
			\item A Rabin automaton recognising $"L_\F"$.
		\end{enumerate}
	The problem of deciding whether $\cmF\geq k$ (or equivalently, $\indmF \geq k$) is $\NPc$.
	\end{theorem}

	\begin{proof}
		We first remark that by Theorem~\ref{theorem : EquivRabin-ChromaticMemory}, the size of a minimal "arena-independent memory" for $\F$ coincides with  the size of a minimal Rabin automaton for $L_\F$.
		
		We show that for the three representations the problem is in $\NP$. By Proposition~\ref{Prop: equivalence Rabin in P} we can find in non-deterministic polynomial time a minimal Rabin automaton recognising $"L_\F"$, if we are given as input a Rabin automaton for $"L_\F"$, (therefore, also if we are given a parity automaton as input). Since the Zielonka tree $\ZF$ allows us to produce in polynomial time a parity automaton for $L_\F$,  $\indmF$ can be computed in non-deterministic polynomial time in $|\ZF|$.
		
		The $\NP$-hardness part has been proven in Theorem~\ref{theorem: MinimisingRabinNP-comp} if the input is a Rabin automaton. If the input is the "Zielonka tree" $\ZF$ or a parity automaton, we show that the reduction presented in Section~\ref{subsec: MinimisationRabin} can also be applied heree. Given a "simple" "undirected" "graph" $G=(V,E)$, we consider the \kl(condition){Muller condition} over $V$ given by $\F_G = \{ \{v,u\}\subseteq V \; : \; (v,u)\in E \}$.
		It is verified that $"L_G" = L_{\F_G}$, and by Lemma~\ref{lemma : RabinForLG = ChromaticNumberG} the size of a minimal Rabin automaton for $"L_G"$ (and therefore the size of a minimal "arena-independent memory" for $\F$) coincides with the "chromatic number" of $G$. It is enough to show that we can build the "Zielonka tree" $\Z_{\F_G}$ and a parity automaton for $L_{\F_G}$ in polynomial time in the size of $G$. The Zielonka tree of  $\F_G$ is the following one: for each pair of vertices $\{v,u\} \subseteq V$ such that $(v,u)\in E$, consider the "tree" $T_{\{v,u\}} = \Big\langle \{v,u\}, \big\langle \langle \{v\} , \langle \emptyset \rangle  \rangle, \langle \{u\}, \langle \emptyset\rangle\rangle \big\rangle \Big\rangle $ (of height $2$). Then, the Zielonka tree of $\F_G$ is given by:
		\[ \Z_{\F_G} = \langle V , \langle T_{e_1}, \dots, T_{e_k}\rangle \rangle, \]
		where $e_1, \dots e_k$ is an enumeration of the unordered pairs of vertices forming an edge. We can build this tree in $\O(|V| + |E|)$.	
		This Zielonka tree gives us a parity automaton for $L_G$ in linear time (Proposition~\ref{prop: minimalParityZielonka}).	
	\end{proof}

\subsection{Chromatic memories require more states than general ones}\label{subset: KopczynskiConjecture}
In his PhD Thesis \cite{Kopczynski2006Half,Kopczynski2008PhD}, Kopczyński raised the question of whether for every winning condition its "general memory requirements" coincide with its "chromatic memory requirements". In this section we provide an example of a game that Eve can win, but she can use a strictly smaller "memory structure" to do so if she is allowed to employ a general memory (non-chromatic). 

\begin{proposition}\label{prop : chromatic_strictly-greater}
	For each integer $n\geq 2$, there exists a set of colours $\GG_n$ and a \kl(condition){Muller condition} $\F_n$ over $\GG_n$ such that for any $\F_n$-game won by Eve, she can win it using a "memory" of \kl(memory){size} $2$, but there is an $\F_n$-game $\G$
	where Eve needs a "chromatic memory" of size $n$ to win.
	Moreover, the game $\G$ can be chosen to be "$\varepsilon$-free".
\end{proposition}

\begin{proof}
	Let $\GG_n=\{1,2, \dots, n\}$ be a set of $n$ colours, and let us define the Muller condition $\F_n$ as:
	\[ \F_n = \{ A\subseteq \GG_n \; : \; |A| = 2\}. \]
	
	The "Zielonka tree" of $\F_n$ is depicted in Figure \ref{Fig : ZielonkaTree}.
	
	\begin{figure}[ht]
		\centering
		\begin{tikzpicture}[square/.style={regular polygon,regular polygon sides=4}, align=center,node distance=2cm,inner sep=3pt]
		
		\node at (0,2.6) [draw, rectangle, minimum height=0.9cm, minimum width = 1.5cm] (R) {$1,2, \dots, n$};
		
		\node at (-3,1.3) [draw, ellipse,minimum height=0.8cm, minimum width = 1.3cm] (0) {$1,2$};
		\node at (-1,1.3) [draw, ellipse, ,minimum height=0.8cm, minimum width = 1.3cm] (1) {$1, 3$};
		\node at (0.6,1.3)  (dots1) {$\dots$};
		\node at (3,1.3) [draw, ellipse, ,minimum height=0.8cm, minimum width = 1.3cm] (2) {$n-1, n$};

		\node at (-3.5,0) [draw, square,minimum width=0.9cm] (00) {$1$};
		\node at (-2.5,0) [draw, square,minimum width=0.9cm] (01) {$2$};
		\node at (-1.5,0) [draw, square,minimum width=0.9cm] (10) {$1$};
		\node at (-0.5,0) [draw, square,minimum width=0.9cm] (11) {$3$};
		\node at (0.9,0)  (dots2) {$\dots$};
		\node at (2.5,0) [draw, rectangle,minimum width=0.9cm, minimum height=0.64cm] (20) {$n - 1$};
		\node at (3.7,0) [draw, square,minimum width=0.9cm] (21) {$n$};
		
		\draw   
		(R) edge (0)
		(R) edge (1)
		(R) edge (2)
		
		(0) edge (00)
		(0) edge (01)
		(1) edge (10)
		(1) edge (11)
		(2) edge (20)
		(2) edge (21);
		
		\end{tikzpicture}
		\caption{Zielonka tree for the condition $\F_{n}=\{ A\subseteq \{1,2,\dots,n\} \: : \: |A|=2\}$. Square nodes are associated with rejecting sets ($A\notin \F_n)$ and round nodes with accepting ones ($A\in \F_n)$.}
		\label{Fig : ZielonkaTree}
	\end{figure}
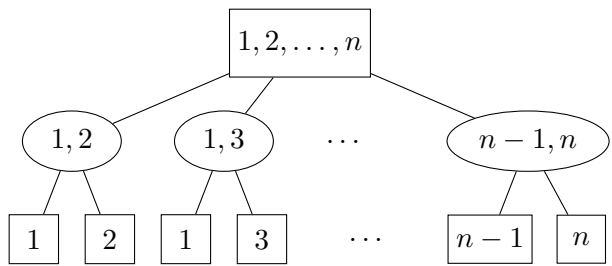
	
	The characterisation of the "memory requirements" of Muller conditions from Proposition~\ref{prop: optimalMemoryDJW} gives that $\mFn = 2$.
	
	On the other hand, the language $"L_{\F_n}"$ associated to this condition coincides with the language $"L_G"$ (defined in Section~\ref{subsec: MinimisationRabin}) associated to a "graph" $G$ that is a clique of size $n$. By Lemma~\ref{lemma : RabinForLG = ChromaticNumberG}, the size of a minimal Rabin automaton recognising $L_{\F_n}$ (and therefore, by Theorem~\ref{theorem : EquivRabin-ChromaticMemory}, the "chromatic memory requirements" of $\F_n$) coincides with the "chromatic number" of $G$. Since $G$ is a clique of size $n$, its "chromatic number" is $n$.
\end{proof}
		
	We give next a concrete example of a game where $\mF < \cmF$.
	\begin{example}\label{example:gameConjKopczynski}
	We consider the set of colours $\GG= \{a,b,c\}$ and the Muller condition from the previous proof: $ \F = \{ A\subseteq \GG \; : \; |A| = 2\}.$ As we have discussed, $\mF = 2$.
	
	 We are going to define a game $\G$ won by Eve for which a minimal "chromatic memory" providing a "winning strategy" cannot have less than $3$ states.
	
	Let $\G$ 
	 be the $\F$-game shown in Figure~\ref{Fig : GameG}. Round vertices represent vertices controlled by Eve, and square ones are controlled by Adam.

	\begin{figure}[ht]
		\centering 
		\begin{tikzpicture}[square/.style={regular polygon,regular polygon sides=4}, state/.style={circle,draw, minimum size=0.7cm}, align=center,node distance=2cm,inner sep=0pt]
		
		\node at (0,0) [state, initial, square, minimum size=0.9cm] (0) {$v_0$};
		\node at (1.75,2) [state] (1) {$v_1$};
		\node at (1.75,0) [state] (2) {$v_2$};
		\node at (1.75,-2) [state] (3) {$v_3$};
		
		\node at (3.5,2.6) [state] (a1) {};
		\node at (3.5,1.4) [state] (b1) {};
		
		\node at (3.5,0.6) [state] (b2) {};
		\node at (3.5,-0.6) [state] (c2) {};
		
		\node at (3.5,-1.4) [state] (a3) {};
		\node at (3.5,-2.6) [state] (c3) {};

		\path[->] 
		(0)  edge [] 	node[above] {$a$ }   (1)
		(0)  edge [] 	node[above] {$a$ }   (2)
		(0)  edge []  node[left] {$a$ }   (3)
		
		(1)  edge [in=170,out=50]  node[above, pos=0.4] {$a$ }   (a1)
		(a1)  edge [in=10,out=210]  node[above, pos=0.4] {$a$ }   (1)
		(1)  edge [in=180,out=-50] 	node[below] {$b$ }   (b1)
		(b1)  edge [in=-10,out=150] 	node[above, pos =0.3] {$b$ }   (1)
		
		(2)  edge [in=170,out=50]  node[above, pos=0.4] {$b$ }   (b2)
		(b2)  edge [in=10,out=210]  node[above, pos=0.4] {$b$ }   (2)
		(2)  edge [in=180,out=-50] 	node[below] {$c$ }   (c2)
		(c2)  edge [in=-10,out=150] 	node[above, pos =0.3] {$c$ }   (2)
		
		(3)  edge [in=170,out=50]  node[above, pos=0.4] {$a$ }   (a3)
		(a3)  edge [in=10,out=210]  node[above, pos=0.4] {$a$ }   (3)
		(3)  edge [in=180,out=-50] 	node[below] {$c$ }   (c3)
		(c3)  edge [in=-10,out=150] 	node[above, pos =0.3] {$c$ }   (3);				
		
		\end{tikzpicture}
		\caption{The game $\G$.}
		\label{Fig : GameG}
	\end{figure}
	
	Suppose that there is a "chromatic memory structure" for $\G$ of size 2, $\M_\G=(M=\{m_0, m_1\}, m_0, \mu)$, together with a function $\nextmove:M\times V_E \rightarrow E$ setting a "winning strategy" for Eve. For winning this game, the number of different colours produced infinitely often has to be exactly two,
	so the transitions of the memory must ensure that, whenever one outgoing edge from $v_i$ is chosen, when we come back  to $v_i$ (after reading two colours in $\GG$) we have changed of state in the memory. Therefore, for each colour $x\in \{a,b,c\}$ we must have that $\mu(m_0,x)=\mu(m_{1},x)$, what implies that there are two different colours (we will suppose for simplicity that they are $a$ and $b$) such that $\mu(m_j,a)=\mu(m_{j},b)=m$, for $j\in \{1,2\}$ and a state $m\in \{m_0, m_1\}$. However, this implies that whenever a partial play ends up in $v_1$, the memory will be in the state $m$, and Eve will always choose to play the edge given by $\nextmove(m,v_1)$, producing a loosing play.		
\end{example}

	\section{Conclusions and open questions}
	In this work, we have fully characterised the "chromatic memory requirements" of \kl(condition){Muller conditions}, proving that "arena-independent" memory structures for a given Muller condition correspond to "Rabin" automata recognising that condition. We have also answered several open questions concerning the "memory requirements" of Muller conditions when restricting ourselves to "chromatic memories" or to "$\varepsilon$-free" games.
We have proven the $\NP$-completeness of the minimisation of transition-based Rabin automata and that we can minimise parity automata recognising \kl(language){Muller languages} in polynomial time, advancing in our understanding on the complexity of decision problems related to transition-based automata.

The question of whether we can minimise transition-based parity or Büchi automata in polynomial time remains open. The contrast between the results of Abu Radi and Kupferman \cite{AbuRadiKupferman19Minimizing, AbuRadiKupferman20Canonicity}, showing that we can minimise GFG transition-based co-Büchi automata in polynomial time and those of Schewe \cite{Schewe20MinimisingGFG}, showing that minimising GFG state-based co-Büchi automata is $\NPc$; as well as the contrast between \Cref{theorem: MinimisingRabinNP-comp} and \Cref{prop : MinimisationParityPoly}, make of this question a very intriguing one.

Regarding the memory requirements of games, we have shown that forbidding $\varepsilon$-transitions might cause a reduction in the memory requirements of Muller conditions. However, the question raised by Kopczyński in \cite{Kopczynski2006Half} remains open: are there prefix-independent winning conditions that are half-positional when restricted to $\varepsilon$-free games, but not when allowing $\varepsilon$-transitions?

\subsection*{Acknowledgements}
I would like to thank Alexandre Blanché for pointing me to the chromatic number problem. 
I also want to thank Bader Abu Radi, Thomas Colcombet, Nathanaël Fijalkow, Orna Kupferman, Karoliina Lehtinen and Nir Piterman for interesting discussions on the minimisation of transition-based automata, a problem introduced to us by Orna Kupferman.
Finally, I warmly thank Thomas Colcombet, Nathanaël Fijalkow and Igor Walukiewicz for their help in the preparation of this paper, their insightful comments and for introducing me to the different problems concerning the memory requirements studied in this paper.		
	
\bibliographystyle{alpha} 
\bibliography{bibMinimisation}
	
\end{document}